\documentclass[11pt, letterpaper]{article}
\title{Near Optimal Hardness of Approximating $k$-CSP}
 \author{
 Dor Minzer\thanks{Department of Mathematics, Massachusetts Institute of Technology, Cambridge, USA. Supported by NSF CCF award 2227876 and 
 NSF CAREER award 2239160.}
 \and
 Kai Zhe Zheng\thanks{Department of Mathematics, Massachusetts Institute of Technology, Cambridge, USA. Supported by the NSF GRFP DGE-2141064.}
 }
\date{\vspace{-5ex}}
\usepackage{fullpage}
\usepackage{amsthm}
\usepackage{amsmath,amssymb,amsfonts,nicefrac}
\usepackage{xspace}
\usepackage{color}
\usepackage{url}
\usepackage{hyperref}
\usepackage{stackengine}
\usepackage{bm}
\usepackage{bbm}
\usepackage{times}
\usepackage{thmtools,thm-restate}
\usepackage{enumitem}
\usepackage[capitalise,nameinlink,noabbrev]{cleveref}

\newcommand{\qbin}[2]{\begin{bmatrix}{#1}\\ {#2}\end{bmatrix}_2}

\DeclareMathOperator{\codim}{codim}

\DeclareMathOperator{\Gras}{{\sf Grass}}

\DeclareMathOperator{\val}{{\sf val}}

\DeclareMathOperator{\spa}{span}

\DeclareMathOperator{\rank}{rank}
\DeclareMathOperator{\im}{im}

\newcommand{\GapLin}{{\sf Gap3Lin}}
\newcommand\card[1]{\left|{#1}\right|}
\newcommand{\Lin}{{\sf 3Lin}}
\newcommand{\NP}{{\sf NP}}

\newcommand{\mc}{\mathcal}
\newcommand{\Eq}{{\sf Eq}}

\newcommand{\Ug}{\mathcal{U}_{{\sf good}}}
\newcommand{\Ql}{\mathcal{Q}_{{\sf lucky}}}
\newcommand{\Qs}{\mathcal{Q}_{{\sf smooth}}}
\newcommand{\Us}{\mathcal{U}_{{\sf sat}}}

\newcommand{\Ws}{W_{{\sf second}}}

\newcommand{\A}{\mathcal{A}}
\newcommand{\T}{\mathcal{T}}

\newcommand{\B}{\mathcal{B}}
\newcommand\inner[2]{\langle{#1},{#2}\rangle}
\newcommand{\norm}[1]{\left\lVert#1\right\rVert}
\newcommand{\Zoom}{{\sf Zoom}}
\newcommand{\Grass}{{\sf Grass}}

\newcommand{\E}{\mathop{\mathbb{E}}}
\newcommand{\Ff}{\mathbb{F}}
\newcommand{\D}{\mathcal{D}}

\newcommand{\U}{\mathcal{U}}

\newcommand{\Rcal}{\mathcal{R}}
\newcommand{\Lcal}{\mathcal{L}}
\newcommand{\Cl}{\textsf{Clique}}

\newcommand{\ind}{\mathbbm{1}}

\newcommand{\eps}{\varepsilon}
\renewcommand{\epsilon}{\eps}

\renewcommand\leq{\leqslant}
\renewcommand\geq{\geqslant}

\theoremstyle{plain} 
\newtheorem{theorem}{Theorem}[section]
   \newtheorem{thm}{Theorem}[section]

   \newtheorem{lemma}[thm]{Lemma}
   
   \newtheorem{claim}[thm]{Claim}
   
   \newtheorem{definition}[thm]{Definition}

\begin{document}
\maketitle
\begin{abstract}
We show that for every $k\in\mathbb{N}$ 
and $\eps>0$, for large enough alphabet 
$R$, given a $k$-CSP with
alphabet size $R$, it is NP-hard to 
distinguish between the case that there is an assignment satisfying at least
$1-\eps$ fraction of the constraints, and
the case no assignment satisfies more than
$1/R^{k-1-\eps}$ of the constraints. 
This result improves upon prior work of [Chan, Journal of the ACM 2016], who showed the same result
with weaker soundness of $O(k/R^{k-2})$, and 
nearly matches the trivial 
approximation algorithm that finds an assignment satisfying 
at least $1/R^{k-1}$ fraction of the constraints.

Our proof follows the approach 
of a recent work by the authors, wherein the above result is proved for $k=2$. Our main new ingredient is a counting lemma for hyperedges between pseudo-random sets in the 
Grassmann graphs, which may be of independent interest.
\end{abstract}

\section{Introduction}
The PCP theorem~\cite{FGLSS, AroraSafra, ALMSS} is a cornerstone result in the theory of computation 
with applications and connections to hardness of approximation, cryptography, pseudorandomness, analysis and more. It has many equivalent formulations, and for this paper it will be convenient to state it in the language of constraints satisfaction problems (CSPs in short).
\begin{definition}
    For an integer $k$, an instance $\Psi$ of $k$-CSP consists of a $k$-uniform hypergraph $G = (V, E)$, an alphabet $\Sigma$, and a set of constraints $\{\Phi_e \}_{e \in E}$. Each constraint is a map $\Phi_e: \Sigma^k \to \{0, 1\}$, where the set of tuples $\Phi_e^{-1}(1)$
    is thought of as assignments to the vertices of $e$ that satisfy the constraints on $e$.
\end{definition}
The arity parameter $k$ and the alphabet size parameter $|\Sigma|$ will be important for
the purposes of this paper. Given a $k$-CSP instance $\Psi$ we define its value as
\[
        \val(\Psi) = \max_{A: V \to \Sigma} \frac{|\{e = (v_1, \ldots, v_k) \; | \; \Phi_e(A(v_1), \ldots, A(v_k))=1\}|}{|E|}.
\]

Using the language of CSPs, the basic PCP theorem~\cite{FGLSS, AroraSafra, ALMSS} combined with 
Raz's parallel repetition theorem~\cite{Raz} 
gives the following result:
\begin{thm} \label{thm: basic 2 query}
    There exists $\gamma > 0$ such that for sufficiently large $R$, given an instance $\Psi$ of $2$-CSP with alphabet size $R$, it is NP-hard to distinguish the following two cases:
    \begin{itemize}
        \item YES case: ${\sf val}(\Psi) = 1$. 
        \item NO case: ${\sf val}(\Psi) \leq \frac{1}{R^{\gamma}}$.
    \end{itemize}
\end{thm}

Let $k>2$. Since any $2$-CSP can be phrased as a
$k$-CSP, it is clear that the above hardness result
applies to $k$-CSPs as well. However, one expects a
stronger hardness result to hold. Indeed, if $\Psi$ is a $2$-CSP, then taking all subsets of $k/2$ constraints from $\Psi$ and for each one putting the $k$-ary constraint 
corresponding to their ``and'' yields a $k$-CSP whose value is $\val(\Psi)^{k/2}$. Thus, \cref{thm: basic 2 query} implies:
\begin{thm} \label{thm: basic k query}
    There exists $\gamma > 0$ such that for sufficiently large $R$, for any even integer $k$, given an instance $\Psi$ of $k$-CSP with alphabet size $R$, it is NP-hard to distinguish the following two cases:
    \begin{itemize}
        \item YES case: ${\sf val}(\Psi) = 1$. 
        \item NO case: ${\sf val}(\Psi) \leq \frac{1}{R^{\gamma (k-1)}}$.
    \end{itemize}
\end{thm}

For many hardness of approximation applications, one only needs the soundness to be a vanishing function of the alphabet size, in which case results such as~\cref{thm: basic 2 query} and~\cref{thm: basic k query} are sufficient. For certain cases, however, one additionally requires the soundness to be small \emph{relative to} the alphabet size~\cite{laekhanukit2014parameters,KS,LeeManurangsi}. 
%This feature can be especially important for obtaining quantitatively strong hardness of approximation results. That is, a better tradeoff can sometimes lead to a better hardness of approximation factor. 
Optimizing the alphabet-soundness tradeoff for $k$-query PCPs is the focus of this paper and quantitatively, one can think of improving this tradeoff as increasing the value of $\gamma$ in \cref{thm: basic k query}. For this task, one should think of $k$ as a fixed constant, and the goal is to obtain larger, concrete values of $\gamma$ for \cref{thm: basic k query} with this arity $k$.

It is clear that the best we can hope in \cref{thm: basic k query} is $\gamma = 1-o(1)$. Indeed, \emph{any} $k$-CSP of alphabet size $R$ has value at least $\frac{1}{R^{k-1}}$, as obtained by a randomly chosen assignment. On the other hand, \cref{thm: basic k query} only achieves a small but absolute value $\gamma > 0$, leaving a large gap between its assertion and what is potentially possible. Strictly speaking, this implicit value of $\gamma$ is still the best known
result to date.

If one is willing to relax the completeness parameter slightly, namely require the value in the YES case to only be close to $1$ (as opposed to exactly $1$), then significantly better values for $\gamma$ in~\cref{thm: basic 2 query} are known. Khot and Safra~\cite{KS} showed that one may take 
$\gamma = 1/6-o(1)$, and this was subsequently improved by Chan~\cite{Chan}, who showed that one may take 
$\gamma = 1/2-o(1)$. This result was further improved recently in~\cite{MZ24}, who achieved the near optimal $\gamma = 1-o(1)$. 

In the same fashion, for $k>2$ stronger forms of~\cref{thm: basic k query} are known to hold, again with imperfect completeness. In this regime, Chan~\cite{Chan} obtains soundness on the orders of $\frac{k}{R^{k-2}}$ if $3 \leq k < R$ (which is the most relevant regime for us, as we think of $k$ as a small constant and $R$ as being a large constant) and $\frac{k}{R^{k-1}}$ for $R \leq k$. In terms of $\gamma$, this roughly translates to $\gamma = \frac{k-2}{k-1}$, which approaches $1$ for sufficiently large $k$, but is still bounded away from the potentially best possible when $k$ is thought of as constant. This gap is even more pronounced for small constant values of $k$ such as $3$ or $4$. In terms of the best approximation ratio possible for $k$-CSP, it is known that approximating the value of $k$-CSP within factor $O_k(\log R/R^{k-1})$ is NP-hard \cite{MNT,KSaket}. The downside of this result is that, when stated as a gap problem, the completeness parameter (the value in the YES case) is small and is a vanishing function of the alphabet size $R$. 
Our main result is a hardness result that achieves a slightly weaker approximation ratio than the result of~\cite{MNT,KSaket}, but has completeness close to $1$.
% Even though these works give optimal NP-hardness of approximation for $k$-CSP with alphabet size $R$, it is still desirable to improve upon their completeness, while maintaining an optimal gap to soundness for applications other NP-hardness results. Our work does precisely this, and our main theorem is the following.

\begin{thm} \label{thm: main}
    For all $\eps, \delta > 0$ and arity $k$, and sufficiently large alphabet size $R$, it is NP-hard to distinguish between the following two cases given a $k$-CSP $\Psi$ with alphabet size $R$ :
    \begin{itemize}
        \item YES case: $\val(\Psi) \geq 1 - \delta$,
        \item NO case: $\val(\Psi) \leq \frac{1}{R^{(1-\eps)(k-1)}}$.
    \end{itemize}
\end{thm}
The fact that the value in the YES case is close to $1$ and that the alphabet-soundness tradeoff is nearly optimal leads to improved hardness results for $k$-Dimensional Matching and related problems.

\paragraph{Application to $k$-Dimensional Matching}
By plugging our improved $k$-CSP hardness result into the reduction from \cite{LST}, we obtain improved hardness of approximation for \emph{$k$-Dimensional matching}.

\begin{definition}
A $k$-Dimensional Matching instance $\Pi = (V, E)$ consists of a $k$-partite hypergraph $G = (V, E)$ with hyperedges of size at most $k$. Its value, $\val(\Pi)$, is the size of the maximum matching of $G$, i.e.\, the largest set of edges in $E$ such that no two edges in $E$ intersect.
\end{definition}

\cref{thm: main} yields the following hardness of approximation result for $k$-Dimensional Matching, improving upon the hardness factor of $\frac{k}{12 + \eps}$ due to \cite{LST}.

\begin{thm} \label{thm: k dim matching}
 For any constant $\eps > 0$, there is a sufficiently large $k$ for which the following holds. Given a $k$-Dimensional Matching instance $\Pi$, no polynomial-time algorithm can approximate $\val(\Pi)$ within a factor of $\frac{k}{8+\eps}$ unless NP is contained in BPP.
\end{thm}

As a corollary, we also get hardness of approximation for a number of other problems. We refer the reader to \cite{LST} for the definition and further discussion on these problems. 
\begin{thm} \label{thm: k dim coro}
  For every constant $\eps > 0$, there is a sufficiently large $k$ such that, unless NP is contained in BPP, there is no polynomial time $\frac{k}{8+\eps}$-approximation algorithm for any of the problems $k$-Set Packing, $k$-Matchoid, $k$-Matroid Parity, or $k$-Matroid Intersection.
\end{thm}
The prior best hardness of approximation factor for all of these problems was $\frac{k}{12+\eps}$, also due to \cite{LST}. We remark
that all of these results require a version of~\cref{thm: main} with respect to $k$-partite $k$-CSPs over regular hypergraphs, and we show how to obtain such result using standard techniques.

\subsection{PCP Construction Overview}
Our reduction from an NP-hard problem to a CSP in \cref{thm: main} closely follows the reduction of \cite{MZ24}. For technical purposes, it is convenient to set the arity of the CSP to $k+1$, rather than $k$. Furthermore, 
it will be helpful to view our reduction to $(k+1)$-CSP as a $(k+1)$-query PCP construction, and we use the language of PCPs henceforth. The completeness of the PCP becomes the value on the YES case of \cref{thm: main} and the soundness becomes the value in the NO case of \cref{thm: main}.

Our $(k+1)$-query PCP is constructed by composing an outer PCP and an inner PCP. 
The outer PCP is exactly the same as that of \cite{KMS,DKKMS1,DKKMS2,KMS2, MZ24}. We give a brief overview of this construction below and refer the reader to~\cite{MZ24} for a 
more detailed discussion.
%and can be thought of as a $2$-Prover-$1$-Round-Game. We refer the reader to the introduction of \cite{MZ24} for background on $2$-Prover-$1$-Round-Games (we note that they are equivalent to $2$-CSPs or $2$-query PCPs). 
Our inner PCP is based on a higher query generalization of Grassmann consistency test of \cite{MZ24}, and our test uses $(k+1)$-queries rather than $2$ (and in return gains a power $k$ improvement in the exponent soundness). The design and analysis of this test is the main new ingredient in our construction. Once we have the outer PCP and  the inner we compose them in similar way to~\cite{MZ24}.

\subsubsection{Outer PCP}
The outer PCP is best described using the following $2$-Prover-$1$-Round Game (see~\cite{MZ24} for a discussion about $2$-Prover-$1$-Round Game and their equivalence to $2$-CSPs). It is based on the $\GapLin$ problem. An instance of this problem is a system of linear equations $\Eq$ over a set of variables $X$, each equation depending on $3$-variables in $X$, and the goal is to assign the variables values from $\Ff_2$ so as to satisfy as many of the equations as possible. By H\r{a}stad's result \cite{Hastad}, for all $\eta>0$ it is NP-hard to distinguish between the case there is an assignment satisfying $(1-\eta)$-fraction of the constraints, and the case every assignment satisfies at most $(1/2+\eta)$-fraction of the equations.

Given a $\GapLin$ instance $(X,\Eq)$, in our outer PCP the verifier first chooses $J$ equations, $e_1, \ldots, e_J \in \Eq$ uniformly at random. The verifier then sends all $3J$-variables involved in these equations to the first prover. As per the second prover, for each equation $e_i$ independently, the verifier picks all three of its variables with probability $1-\beta$, and otherwise picks one of them randomly. The verifier sends this collection of variables to the second prover. 
The two provers are expected to give consistent assignments to all variables they are sent, and furthermore the assignment of the first prover has to satisfy the equations $e_1, \ldots, e_J$.
This game is known as the smooth parallel repetition of the variables-versus-equation game, and in~\cite{KMS,MZ24} it
was shown to have near perfect completeness in the YES case, and 
soundness $2^{-\Theta(\beta J)}$ in 
the NO case. We refer the reader to~\cite[Section 1.2.4]{MZ24} for more details.\footnote{Strictly speaking, to be able to compose the outer PCP with an inner PCP we also require a feature known as ``advice''. The advice for the second prover can be thought of as a list of uniformly random vectors in the space $\Ff_2^{V}$, where $V$ is the set of variables sent to the second prover. Letting $U$ be the set of variables sent to the first prover, the advice for the first prover can be thought of as a list of vectors in $\Ff_2^U$, which are the result of filling $0$'s in coordinates of $U\setminus V$ of the second prover's advice vectors.}

\subsubsection{Inner PCP}\label{sec:inner_PCP}
Our inner PCP is based on a local test, called the $(k+1)$-query Grassmann consistency test. 
It broadly falls in the category of ``subspace versus subspace tests'', which are ubiquitous in PCP constructions, going back all the way to the original proof of the PCP theorem. These tests include the Raz-Safra Plane versus Point test \cite{RS} and its variants in other dimensions \cite{AroraSafra, MR, KS, BDN, MZ, HKSS}.  Despite all of these prior works, obtaining \cref{thm: main} requires us to consider a new setting wherein the number of queries made is greater than $2$. Furthermore, due to the tight alphabet-soundness tradeoff we are hoping for, we must make these extra queries in a non-trivial manner. At a high level, each extra query needs to nearly beget the power of a fully independent $2$-query test. 

The setting of the Grassmann consistency is as follows. 
Think of $n$ as large and let $\Grass(n, r)$ denote the set of subspaces of $\mathbb{F}_2^n$ of dimension $r$. Suppose that a linear function $f: \Ff_2^n \to \Ff_2$ is encoded via tables $T_1$ and $T_2$ as follows: $T_1$ assigns to each $L\in \Grass(n,2\ell)$ the linear function $T_1[L] = f|_L: L \to \Ff_2$, and $T_2$ assigns to each $R\in \Grass(n,2(1-\delta)\ell)$ the linear function $T_2[R] = f|_R: R \to \Ff_2$. Given two tables $T_1, T_2$ which assign linear functions to subspaces (not necessarily as above), how can one test if they truly come from a global linear function $f$ as above? An obvious $2$-query tester proceeds as follows:
\begin{itemize}
    \item Choose $R$ of dimension $2(1-\delta)\ell$ uniformly and $L \supseteq R$ of dimension $2\ell$ uniformly,
    \item Accept if and only if $T_1[L]|_R \equiv T_2[R]$.
\end{itemize}
It is clear that the test has completeness $1$: if the two tables truly come from the same global function, then the test always accepts. As for the soundness, it is shown in~\cite{MZ24} to be roughly $2^{-2(1-1000\delta)\ell}$, albeit under a non-traditional sense of soundness. Roughly
speaking, their result
states that there are subspaces $Q,W$ of constant dimension and codimension respectively, such that $T_1,T_2$ agree with a global linear function on non-trivial fraction of subspaces $L$ such that $Q\subseteq L\subseteq W$.\footnote{It can be shown that the more standard notion of soundness fails, making this non-standard version necessary. See~\cite{MZ24}.}
Since the tables have alphabet size $2^{2\ell}$ (coming from the number of linear functions over a $2\ell$ dimension subspace over $\Ff_2$), these soundness and alphabet size parameters match the nearly optimal tradeoff of \cref{thm: main} in the $2$-query case.

Getting a $(k+1)$-query test with near optimal soundness requires some more thought. Using the analysis of \cite{MZ24} in a black-box way seems insufficient, as it only allows one to repeat the above tester $(k+1)/2$ times, and thus  amplifies the soundness to just $2^{-2(1-1000\delta) \ell \cdot (k+1)/2}$. This is quadratically off from the soundness in~\cref{thm: main}. Therefore, to achieve the desired soundness for \cref{thm: main}, we need to make the queries in a more dependent manner while not compromising the soundness. Specifically, we are hoping to match the soundness of $k$ independent trials of the $2$-query test while making only half the number of queries. We achieve this via the following test:

\begin{itemize}
    \item Choose $R \in \Grass(n, 2(1-\delta)\ell)$ uniformly and $L_1, \ldots, L_k \in \Grass(n, 2\ell)$ such that $R \subseteq L_i$ for all $i \in [k]$ uniformly and independently.
    \item Accept if and only if $T_1[L_i]|_R \equiv T_2[R]$ for all $i \in [k]$.
\end{itemize}
At a high level, one can think of the test as performing $k$ repetitions of the $2$-query test, but reusing the subspace $R$ across each trial. By reusing $R$, we are able to save on the number of queries, and furthermore, we show that this change does not sacrifice soundness. Specifically we show that if the above test has soundness $2^{-2(1-1000\delta)\cdot k\ell}$, under the same non-traditional sense of soundness from \cite{MZ24}. Our analysis involves Fourier analysis and makes use of recent global hypercontractivity results over the bilinear scheme \cite{EvraKL}.

\subsubsection{Our Main Analytical Statement}
At the heart of our analysis lies an upper bound on the number of hyperedges 
between two \emph{pseudo-random sets} in 
Grassmann graphs, which may be of independent interest. 
\begin{definition}
  We say that a set $S\subseteq   \Grass(n, \ell)$ is \textit{$(r, \epsilon)$-pseudo-random} if for all subspaces $Q \subseteq W \subseteq \Ff_2^n$ satisfying $\dim(Q) + \codim(W) = r$, we have $\Pr_{L\in \Grass(n,\ell)}[L\in S~|~Q\subseteq L\subseteq W]  \leq \epsilon$.
\end{definition}
With this definition we have the following result:
\begin{lemma}[Informal version of~\cref{lm: pseudorandom edges}]
Suppose that $\mathcal{L}\subseteq \Grass(n,2\ell)$ is $(r,\eps)$-pseudo-random
and $\mathcal{R}\subseteq \Grass(n,2(1-\delta)\ell)$ is a set with density $\beta$. Then
\[
\Pr_{\substack{R\in \Grass(n,2(1-\delta)\ell)\\ \Grass(n,2\ell)\ni L_1,\ldots,L_k\supseteq R}}[R\in\mathcal{R}, L_1,\ldots,L_k\in \mathcal{L}]
\leq 2^{O_{r,k}(1)}\beta^{1-o(1)}\eps^{k(1-o(1))}+2^{-r\delta\ell+1}.
\]
\end{lemma}
In other words, taking $r$ large so that the $2^{-r\delta\ell+1}$ term is negligible, the lemma says that the probability that $R$ and $L_1,\ldots,L_k$ sampled as above are in $\mathcal{R}$ and 
$\mathcal{L}$ respectively is not
much more than in a random $\mathcal{L},\mathcal{R}$ of densities 
$\eps,\beta$, which would be $\beta \eps^{k}$. 
We remark that there is some analogy between the above lemma and~\cite[Lemma 3.1]{MNT}, which makes a similar in spirit assertion for the noisy hypercube. Their
result however requires a large noise rate, which is ultimately the reason their result has small completeness.

\section{Preliminaries}

\subsection{The Grassmann Domain}
A key component of our result is Fourier analysis over graphs related to the Grassmann domains, which we overview here. Throughout this section, we fix parameters $n$ and $\ell$ with $1 \ll \ell \ll n$.
\subsubsection{Basic Definitions}

We use $\Grass(n, \ell)$ to denote the set of $\ell$-dimensional subspaces $L \subseteq \Ff_2^n$. At times we will have a different vector space $V$ over $\Ff_2$ used as the ambient space, in which case $\Grass(V, \ell)$ denotes the set of $\ell$-dimensional subspaces of $V$. 
The number of $\ell$-dimensional subspaces of $\Ff_2^n$ is denoted by the Gaussian binomial coefficient $\qbin{n}{\ell} = \prod_{i=0}^{\ell-1}\frac{2^n - 2^i}{2^\ell - 2^i}$. Abusing notation, we denote  $\qbin{V}{\ell} = \Grass(V, \ell)$.

\paragraph{Bipartite Inclusion Graphs.} We will often consider the bipartite inclusion graph between $\Grass(n, \ell)$ and $\Grass(n, \ell')$ for $\ell' < \ell$. This graph has an edge between every $L \subseteq \Grass(n,\ell)$ and $L' \in \Grass(n, \ell')$ such that $L \supseteq L'$. We refer to the family of such graphs as \emph{Grassmann graphs}. The edges of these are the motivation for the test used in our inner PCP. While we never refer to this graph explicitly, it will be helpful to have it in mind --- particularly during \cref{sec: grassmann test}.

\paragraph{Zoom-ins and Zoom-outs.}
A feature of the Grassmann domain is that they contain many copies of lower dimensional Grassmann sub-domains. We refer to these sub-domains as zoom-ins and zoom-outs, and they play a large part in the analysis of our inner PCP and final PCP. For subspaces $Q \subseteq W \subseteq \Ff_2^n$, let 
\[
\Zoom[Q,W] = \{L \in \Grass(n, \ell) \; | \; Q \subseteq L \subseteq W\}.
\]
We refer to $Q$ as a zoom-in and $W$ as a zoom-out. When $W = \Ff_2^n$, $\Zoom[Q,W]$ is the zoom-in on $Q$, and when $Q = \{0\}$, $\Zoom[Q,W]$ is the zoom-out on $W$. 

 \subsubsection{Pseudo-randomness over the Grassmann graph}
The notion $(r,\epsilon)$-pseudo-randomness, which captures how much the measure of a set $\mc{L}$ can increase by a zoom-in/zoom-out restrictions of ``size $r$'', will be important for us. Denote $\mu(\mc{L}) = \Pr_{L \in \Grass(n, \ell)}[L \in \mc{L}]$, and for subspaces $Q \subseteq W \subseteq \Ff_2^n$, denote 
$\mu_{Q, W}(\mc{L}) = \Pr_{L \in \Zoom[Q,W]}[L \in \mc{L}]$.
\begin{definition}
  We say that a set $L \subseteq \Grass(n,\ell)$ is \textit{$(r, \epsilon)$-pseudo-random} if for all $Q \subseteq W \subseteq \Ff_2^n$ satisfying $\dim(Q) + \codim(W) = r$, we have $\mu_{Q, W}(\mc{L})  \leq \epsilon$.
\end{definition}

\subsection{The Bilinear Scheme} \label{sec: bilinear scheme}

The bilinear scheme is a domain closely related to the Grassmann domain: instead of considering $\ell$-dimensional subspaces, one considers $n \times \ell$ matrices over $\Ff_2$. To move from the Grassmann domain to the Bilinear scheme we will associate a matrix $M$ with its column span. 
Thus, the matrices of $\Ff_2^{n \times \ell}$ roughly correspond to the subspaces in $\Grass(n, \ell)$ and likewise the matrices of $\Ff_2^{n \times \ell'}$ and subspaces of $\Grass(n, \ell')$. This will be important for us later when we define the analogs of the Grassmann inclusion graphs in the Bilinear scheme. 
%The edges of the bilinear scheme can be defined in a natural and suitable way, but some care is needed to handle matrices with linearly dependent columns, i.e.\ whose span has dimension less than $\ell$. 
%This inconvenience is minor nevertheless, and we are able to translate results from the bilinear scheme to the Grassmann scheme without much loss. 
The main advantage of the bilinear scheme is that it is a Cayley graph and therefore Fourier analysis is more easily applicable, and in the rest of
this section we give some background on it.

\paragraph{The Bilinear Scheme:} 
We denote by $L_2\left(\Ff_2^{n \times 2\ell}\right)$ the space of complex valued functions $F: \Ff_2^{n \times 2\ell} \to \mathbb{C}$ equipped with the inner product $\langle F, G \rangle = \E_{M \in \Ff_2^{n \times 2\ell}}[F(M) \overline{G(M)}]$,
where the distribution taken over $M$ is uniform. 
For $s \in \Ff_2^{n}$ define $\chi_s\colon \mathbb{F}_2^n\to\{-1,1\}$ by $\chi_s(x) = (-1)^{s \cdot x}$. Then, an orthonormal basis of $L_2(\Ff_2^{n \times 2\ell})$ is given by the characters, $\chi_S: \Ff_2^{n \times 2\ell} \xrightarrow[]{} \{-1,1\}$ over all $S = (s_1,\ldots, s_{2\ell}) \in \Ff_2^{n \times 2\ell}$, where
\[
\chi_S(x_1,\ldots, x_{2\ell}) = \prod_{i=1}^{2\ell} \chi_{s_i}(x_i) =  (-1)^{\sum_{i=1}^{2\ell} s_i \cdot x_i}.
\]
Thus, any $F \in L_2(\Ff_2^{n \times 2\ell})$ can be expressed as $F = \sum_{S \in \Ff_2^{n \times 2\ell}}\widehat{F}(S) \chi_{S}$,
where $\widehat{F}(S) = \langle F, \chi_S \rangle$. The level $d$ component of $F$ is given by
\begin{equation} \label{eq: level d decomp}   
F^{=d} = \sum_{S: \; \rank(S) = d} \widehat{F}(S) \chi_{S}.
\end{equation}

If a function $F$ only consists of components up to level $d$, i.e.\ $\widehat{F}(S) = 0$ for all $\rank(S) > d$, then we say $F$ is of degree $d$.

We now describe the analogues of zoom-ins and zoom-outs on the bilinear scheme over $\Ff_2^{n \times 2\ell}$, as well as an analogous notion of $(r, \epsilon)$-pseudo-randomness for Boolean functions over $\Ff_2^{n \times 2\ell}$. 

\begin{definition}
  A zoom-in of dimension $d$ over $\Ff_2^{n \times 2\ell}$ is given by $d$-pairs of vectors $(u_1, v_1), \ldots, (u_r, v_r)$ where each $u_i \in \Ff_2^{2\ell}$ and each $v_i \in \Ff_2^{n}$. Let $U \in \Ff_2^{2\ell \times d}$ and $V \in \Ff_2^{n \times d}$ denote the matrices whose $i$th columns are $u_i$ and $v_i$ respectively. Then the zoom-in on $(U, V)$ is the set of $M \in \Ff_2^{n \times 2\ell}$ such that $MU = V$, or equivalently, $Mu_i = v_i$ for $1 \leq i \leq d$.  
\end{definition}

\begin{definition}
 A zoom-out of dimension $d$ is defined similarly, except by multiplication on the left. Given $X \in \Ff_2^{d \times n}$ and $Y \in \Ff_2^{d \times 2\ell}$, whose rows are given by $x_i$ and $y_i$ respectively, the zoom-out $(X, Y)$ is the set of $M \in \Ff_2^{n \times 2\ell}$ such that $XM = Y$, or equivalently, $x_iM = y_i$ for $1 \leq i \leq d$.   
\end{definition}

Let $\Zoom[(U, V), (X, Y)]$ denote the intersections of the zoom-in on $(U, V)$ and the zoom-out on $(X, Y)$. The codimension of $\Zoom[(U, V), (X, Y)]$ is the sum of the number of columns of $U$ and the number of rows of $X$, which we will denote by $\dim(U)$ and $\codim(X)$. For a zoom-in and zoom-out pair and a Boolean function $F$, we define $F_{(U, V), (X,Y)}: \Zoom[(U,V), (X,Y)] \xrightarrow[]{} \{0,1\}$ to be the restriction of $F$ which is given as 
\[
F_{(U, V), (X,Y)}(M) = F(M) \quad \text{for} \quad M \in  \Zoom[(U,V), (X,Y)].
\]
When $\dim(U) + \codim(X) = d$, we say that the restriction is of size $d$. We define $(d, \epsilon)$-pseudo-randomness in terms of the $L_2$-norms of restrictions of $F$ of size $d$.
Here and throughout, when we consider
restricted functions, the underlying 
measure is the uniform measure
over the corresponding zoom-in and zoom-out
set $\Zoom[(U,V),(X,Y)]$.
\begin{definition} \label{def: bilinear scheme pseudorandom}
We say that an indicator function $F \in L_2(\Ff_2^{n \times 2\ell})$ is \emph{$(d, \epsilon)$-pseudo-random} if for all zoom-in zoom-out combinations $\Zoom[(U, V), (X, Y)]$ such that $\dim(U) + \codim(X) = d$, we have
\[
\norm{F_{(U,V), (X, Y)}}_2^2 
\leq \epsilon.
\]
\end{definition}
We note that for Boolean functions $F$, 
$\norm{F_{(U,V), (X, Y)}}_2^2 
= \E_{M \in \Zoom[(U, V), (X, Y)]}[F(M)]$, and hence the 
definition above generalizes the definition we have for Boolean
functions. 

\section{The PCP Construction}
In this section we describe our PCP construction towards the proof of~\cref{thm: main}.
% Our PCP is constructed by composing an outer PCP with an inner PCP. This technique of composition is ubiquitous in the PCP literature and is present in almost all known PCP constructions. Our outer PCP is presented as a game between a verifier and two provers, or a $2$-Prover-$1$-Round Game, which is equivalently a $2$-query PCP. We then compose this outer PCP with an inner PCP, which can be thought of as a $k+1$-query test, and altogether this yields our composed, $k+1$-query PCP.

\subsection{Hardness of 3LIN}
As a starting point, we use a classical result of H\r{a}stad on the hardness of approximately solving systems of linear equations. To state this result, let us start by formally defining the problem $\Lin$.

\begin{definition}
    An instance of $\Lin$, written as
    $(X, \Eq)$, consists of a set of variables $X$ and 
    a set of linear equations $\Eq$ over $\mathbb{F}_2$. 
    Each equation in $\Eq$ depends on exactly three variables in $X$, each variable appears in at most $10$ equations, and any two distinct equations in $\Eq$ share at most a single variable. 
\end{definition}
The goal in $\Lin$ is to find an assignment 
$A\colon X\to\Ff_2$ satisfying as many of the equations
in $\Eq$ as possible. The maximum fraction of equations that 
can be satisfied is called the \textit{value} of the instance.
For $0<s<c\leq 1$, $\GapLin[c,s]$ is the promise
problem wherein the input is an instance $(X,\Eq)$ of $\Lin$
promised to either have value at least $c$ or at most $s$, 
and the goal is to distinguish between these two cases. 

The conditions that two equations in 
$\Eq$ share at most a single variable and the bound on the number of occurrences of each variable are often not included 
in the definition of $\Lin$.
However, combining the classical result of H\r{a}stad~\cite{Hastad} with elementary reductions (see~\cite{dor_thesis,MZ24}) one has the following result:
\begin{theorem} \label{th: 3lin hardness}
There exists $s<1$ such that for every constant $\eta > 0$, $\GapLin\left[1-\eta, s\right]$ is $\NP$-hard.
\end{theorem}

\subsection{The Outer PCP}\label{sec:outer}  \label{sec: final outer pcp}
The outer PCP we use is the same as the outer PCP of \cite{MZ24}. Fix an instance of $\GapLin[1-\eps_1, 1-\eps_2]$ where $0 < \eps_1 < \eps_2$ and call the instance $(X, \Eq)$, so that the variables are $X$ and the equations are $\Eq$. Using, $(X, \Eq)$ we construct a $2$-Prover-$1$-Round game with arbitrarily small constant value and two additional features called ``smoothness'' and ``advice'', and this game serves as our outer PCP. 
%Due to the additional features, the final construction is somewhat complicated and seemingly unnatural when presented by itself. Behind the scenes though, this construction is actually developed through a gradual evolution from a more basic and intuitive game. 
This construction can be motivated more gradually, but for the sake of conciseness we only describe the final construction below. We refer the reader to~\cite[Section 3]{MZ24} for a more detailed discussion. 

Let $0 \leq \beta < 1$ be a smoothness parameter, $r \in \mathbb{N}$ be an advice parameter, and $J$ be a repetition parameter. Then the game $G^{\otimes J}_{\beta,r}$ proceeds as follows:
\begin{enumerate}
    \item The verifier chooses equations $e_1,\ldots,e_J \in \Eq$ uniformly and independently, and lets $U_i$ be the set of variables in $e_i$.
    \item For each $i$ independently, with probability $1-\beta$, the verifier chooses 
    $V_i = U_i$. With probability $\beta$, the verifier chooses
    $V_i\subseteq U_i$ randomly of size $1$.
    \item For each $i=1,\ldots,J$ independently, the verifier picks a vectors $v_1^{i},\ldots,v_r^{i}\in \Ff_2^V$ 
    uniformly and independently. If $U_i = V_i$ the verifier takes $u_j^{i} = v_j^{i}$ for $j=1,\ldots,r$, and 
    otherwise the verifier takes the vectors $u_1^{i},\ldots,u_r^{i}\in\Ff_2^U$
    where for all $j=1,\ldots,r$, the vector $u_j^{i}$ agrees with $v_j^{i}$ on the coordinate of $V_i$, and is 
    $0$ in the coordinates of $U_i\setminus V_i$.
    \item The verifier sets $U = \bigcup_{i=1}^J U_i$ and $u_j = (u_j^1,\ldots,u_j^J)$ for each $j=1,\ldots r$, and 
    $V = \cup_{i=1}^{J} V_i$ and $v_j = (v_j^1,\ldots,v_j^J)$ 
    for each $j=1,\ldots,r$. The verifier sends $U$ and 
    $u_1,\ldots,u_r$ to the first prover, and $V$ and 
    $v_1,\ldots,v_r$ to the second prover.
    \item The provers respond with assignments to the variables they receive, and the verifier accepts if and only if their 
    assignments agree on $V$ and the assignment to $U$ satisfies
    the equations $e_1,\ldots,e_J$.
\end{enumerate}
In~\cite{KMS,DKKMS1, MZ24} the above game is referred to as the $J$-fold parallel repetition of the smooth variable versus equation game with advice. The completeness and the soundness of $G_{\beta,r}^{\otimes J}$ is established in those works via the parallel repetition theorem~\cite{Raz,Rao}, and we state them below.
\begin{claim} \label{claim:soundness_of_outerpcp}
Given a $\Lin$ instance $(X, \Eq)$, the game $G^{\otimes J}_{\beta,r}$ satisfies the following:
\begin{itemize}
    \item \textbf{Completeness:} If $(X, \Eq)$ has an assignment satisfying at least $(1-\eps)$
of the equations, then the provers can win 
$G_{\beta,r}^{\otimes J}$ with probability at least $(1-J\eps)$.
\item \textbf{Soundness:} If no assignment to $(X, \Eq)$ satisfies more than $(1-\eps)$
of the equations, then the provers can win 
$G_{\beta,r}^{\otimes J}$ with probability at most 
$2^{-\Omega(\eps^2 2^{-r}\beta J)}$.
\end{itemize}
\end{claim}
\begin{proof}
The completeness case is \cite[Claim 3.1]{MZ24}, and the soundness is \cite[Claim 3.2]{MZ24}.
\end{proof}

\paragraph{Viewing the advice as a subspace.} 
As each variable appears in at most $O(1)$
equations, with probability $1-O(J^2/n)$ all variables in $e_1,\ldots,e_J$ are distinct.
In this case, note that the $r$ vectors of advice 
to the second prover, $v_1, \ldots, v_r \in \Ff_2^V$, 
are uniform, and the second prover may consider their span $Q_V$. As for the first prover, the vectors $u_1,\ldots,u_r\in \Ff_2^U$ are not uniformly distributed. Nevertheless,
as shown by the covering property from~\cite{KS,KMS} (and presented later), the distribution of $u_1,\ldots,u_r$ is 
close to uniform over $r$-tuples of vectors from $\Ff_2^U$.
Thus, the first prover can also take their span, call it 
$Q_U$, and think of it as a random $r$-dimensional subspace
of $\Ff_2^U$.

\subsection{The Composed PCP}\label{sec:pcp_construct}
We now describe the composition of the outer PCP from~\cref{sec:outer} and the inner PCP from~\cref{sec:inner_PCP}. We describe this PCP as a CSP, and towards this end we define its vertex sets, alphabets and constraints. Henceforth, we fix $k \in \mathbb{N}$ to be a constant. The composed PCP will have $(k+1)$ queries and the equivalent CSP will have arity $(k+1)$.

Fix an instance $(X, \Eq)$ of $\GapLin$, consider the outer game $G_{\beta,r}^{\otimes J}$ from~\cref{sec:outer}, and let $\U$ denote the set of questions asked to the first prover. Specifically, $\U$ consists of all $J$-tuples of equations $U = (e_1, \ldots, e_J) \in \Eq^{J}$ from the $\GapLin$ instance $(X, \Eq)$. For $e \in \Eq$ let $x_e \in \Ff_2^{X}$ denote the indicator vector on the three variables appearing in $e$.
It will be convenient to only keep the $U = (e_1, \ldots, e_J)$ that satisfy the following properties:
\begin{itemize}
    \item The equations $e_1,\ldots, e_J$ are distinct and do not share variables.
    \item For any $i \neq j$ and pair of variables $x \in e_i$ and $y \in e_j$, the variables $x$ and $y$ do not appear together in any equation in the instance $(X, \Eq)$. 
\end{itemize}

The fraction of $U = (e_1, \ldots, e_J)$ that do not satisfy the above is $O(J^2/n)$ which is negligible for us, and dropping
them from $\U$ will only reduce our completeness by $o(1)$. This will not affect our analysis, and henceforth we will assume that all $U = (e_1, \ldots, e_J)$ satisfy the above properties. 

For each $U$, we will then include a copy of the Grassmann consistency test inside $\Ff_2^U$. As a result, all vertices in the underlying graph of our $(k+1)$-CSP will correspond to subspaces of $\Ff_2^X$.

\subsubsection{The Vertices}
For each question $U = (e_1, \ldots, e_J)$, let $H_U = \spa(x_{e_1}, \ldots, x_{e_J})$, where we recall that $x_{e_i}$ is the vector with ones at coordinates corresponding to variables appearing in $e_i$. By the first property described above we have $\dim(H_U) = J$ and $\dim(\Ff_2^{U}) = 3J$. The vertices of $\Psi$ are: 
\begin{align*}
&\A = \{L \oplus H_U \; | \; U \in \U, L \subseteq \Ff_2^U, \dim(L) = 2\ell, L \cap H_U = \{0\} \},\\
&\B = \{R \; | \; \exists U \in \U, \text{ s.t. } R \subseteq \Ff_2^U, \dim(R) = 2(1-\delta)\ell  \}.
\end{align*}
In words, the vertices on the side $\mathcal{A}$ correspond to $2\ell$-dimensional subspaces of some $\Ff_2^{U}$ over all $U\in \mathcal{U}$. For technical reasons, we require them to intersect
$H_U$ trivially (which is the case for a typical $2\ell$-dimensional space) and add to them the space $H_U$.\footnote{This has the effect of collapsing $L$ and $L'$ such that $L\oplus H_U = L'\oplus H_U$ to a single vertex.}
The vertices on the side $\mathcal{B}$ are all $2(1-\delta)\ell$-dimensional subspaces of $\Ff_2^U$ over all $U\in\U$.

\subsubsection{The Alphabets} 
The alphabets $\Sigma_1, \Sigma_2$ have sizes $|\Sigma_1| = 2^{2\ell}$ and $|\Sigma_2| = 2^{2(1-\delta)\ell}$. For each vertex $L \oplus H_U \in \A$, let $\psi: H_U \xrightarrow[]{} \Ff_2$ denote the function that satisfies the side conditions given by the equations in $U$. Namely, if $e_i \in U$ is the equation $\langle x, x_{e_i} \rangle = b_i$ for $x \in \Ff_2^U$, then $\psi(x_{e_i}) = b_i$. We say a linear function 
$f\colon L\oplus H_U\to\Ff_2$ satisfies the side conditions of $U$ if $f|_{H_{U}}\equiv \psi$. 
In this language, for a vertex $L \oplus H_U$ we identify $\Sigma_1$ with
\[
\{f : L \oplus H_U \xrightarrow[]{} \Ff_2 \; | \; f \text{ is linear function satisfying the side conditions of $U$} \}.
\]
As $L \cap H_U = \{0\}$ and $\dim(L) = 2\ell$, it is easy to see that the above set indeed has size $2^{2\ell}$. Analogously, for each right vertex $R\in\B$ we identify $\Sigma_2$ with $\{f : R \xrightarrow[]{} \Ff_2 \; | \; f \text{ is linear}\}$.

\subsubsection{The Constraints} \label{sec: constraint graph}
To define the constraints, we first need the following relation on the vertices in $A$. Say that $(L \oplus H_U) \sim (L' \oplus H_{U'})$ if 
\[
L + H_U + H_{U'} = L' + H_U + H_{U'}.
\]
Recall that all subspaces above are in $\Ff_2^{X}$ hence the sums and equality above are well defined. In \cite{MZ24}, it is shown that this relation is in fact an equivalence relation, meaning it is reflexive, symmetric, and transitive. As such, the vertices of $A$ are partitioned into equivalence classes, which we denote $[L \oplus H_U]$ and call cliques. Specifically, $[L \oplus H_U]$ is the equivalence class, under $\sim$, containing the vertex $L \oplus H_U$.

\begin{lemma} \label{lm: transitive}
    The relation $\sim$ described above is an equivalence relation. 
\end{lemma}
\begin{proof}
 See~\cite[Lemma 4.1]{MZ24}.
\end{proof}

One feature of the cliques $[L \oplus H_U]$ is that given any two $L \oplus H_U$ and $L' \oplus H'_{U'}$ in the same clique, there is a unique way to go from a linear function over one vertex to a linear function over another vertex. In this sense, the assignments over vertices in the same clique actually contain the same ``information''. 

\begin{lemma} \label{lm: clique extension}
    Suppose $L \oplus H_U \sim L' \oplus H_{U'}$ and that $f: L \oplus H_U \xrightarrow[]{} \Ff_2$ is a linear function satisfying the side conditions. Then there is a unique linear function $f': L' \oplus H_{U'} \xrightarrow[]{} \Ff_2$ that satisfies the side conditions such that there exists a linear function $g:  L + H_U + H_{U'} \xrightarrow[]{} \Ff_2$ satisfying the side conditions (of both $U$ and $U'$) such that
    \[
    g|_{ L \oplus H_U} = f \quad \text{ and } \quad g|_{ L' \oplus H_{U'}} = f'.
    \]
    In words, $g$ is a linear extension of both $f$ and $f'$.
\end{lemma}
\begin{proof}
    Note that there is only one way to extend $f$ to $L + H_U + H_{U'}$ in a manner that satisfies the side conditions given by $U'$. Let this function be $g$. We take $f'$ to be $g|_{L' \oplus H_{U'}}$.
\end{proof}

We now describe the constraints of our $(k+1)$-CSP. Denote by $T_1$ the assignment to $\mathcal{A}$ 
that assigns, to each vertex $L\oplus H_U$, a linear function $T_1[L\oplus H_U]$ on $L\oplus H_U$ satisfying the side conditions, and by $T_2$ the assignment that assigns to each vertex $R\in \mathcal{B}$ a linear function $T_2[R]$ on $R$. 
The constraints of our $(k+1)$-CSP and their weights are given by the following randomized procedure:

\begin{enumerate}
    \item Choose $U$ uniformly at random from $\U$.
    \item Choose $R \subseteq \Ff_2^U$ with $\dim(R) = 2(1-\delta)\ell$ uniformly at random such that $R \cap H_U = \{0\}$.
    \item For each $i \in [k]$, choose $L_i \oplus H_U$ uniformly, such that $\dim(L_i) = 2\ell$, $L_i \cap H_U = \{0\}$, and $L_i \supseteq R$.
    \item For each $i \in [k]$, choose $L'_i \oplus H_{U'_i} \in [L_i \oplus H_U]$ uniformly and extend $T_1[L'_i \oplus H_{U'_i}]$ to $L'_i + H_{U'_i} + H_U$ 
    in the unique manner that respects the side conditions (as in~\cref{lm: clique extension}). 
    Let $\Tilde{T}_1[L_i \oplus H_U]$ be the restriction of this extension to $L_i \oplus H_U$.
    \item The constraint is over the vertices $L'_{i} \oplus H_{U'_i}$ and $R$, and is satisfied if and only if $\Tilde{T}_1[L_i \oplus H_U]|_{R} = T_2[R]$ for all $i \in [k]$.
\end{enumerate}
This finishes the description of our instance $(k+1)$-CSP, $\Psi$. It is 
clear that the running time and instance size is $n^{O(J)}$, and also that the alphabet size is $O(2^{2\ell})$. 
 As is often the case the 
 completeness of the reduction is relatively straightforward, and the soundness is more complicated. In the next section we develop some tools for the soundness analysis.

\section{A $(k+1)$-query Grassmann Consistency Test} \label{sec: grassmann test}
The $(k+1)$-query Grassmann consistency test, which we discuss in this section, is the combinatorial core of the inner PCP. It is similar to the
Grassmann-based PCPs of \cite{KMS, KMS2, MZ24} and in particular can be viewed as a $(k+1)$-query generalization of the \cite{MZ24} Grassmann consistency test.

Fix $U$ as a question to the first prover in the outer PCP game, i.e., the set of variables appearing in $J$ equations of the initial $\GapLin$ instance. Suppose $T_1,T_2$ are tables such that for each $L \in \Grass(U, 2\ell)$ such that $L \cap H_U = \{0\}$, $T_1[L \oplus H_U]: L \oplus H_U \to \Ff_2$ is a linear function satisfying the side conditions $H_U$, and for each $R \in \Grass(U, 2(1-\delta)\ell)$ such that $R \cap H_U = \{0\}$, $T_2[R]: R \to \Ff_2$ is a linear function. The $(k+1)$-query Grassmann consistency test proceeds as follows:
\begin{itemize}
    \item Choose $R \in \Grass(\Ff_2^U, 2(1-\delta)\ell)$ such that $R \cap H_U = \{0\}$ uniformly.
    \item Choose $L_1, \ldots, L_k \in \Grass(\Ff_2^U, 2\ell)$ uniformly such that for each $i \in [k]$, $L_i \cap H_U = \{0\}$ and $R \subseteq L_i$.
    \item Accept if and only if $T_1[L_i\oplus H_U]|_R \equiv T_2[R]$ for all $i \in [k]$.
\end{itemize}

To analyze this test we need the following lemma, which, bounds the probability that a randomly chosen $R \in \Grass(n, 2(1-\delta)\ell)$ along with $k$ randomly chosen $L_1, \ldots, L_k \in  \Grass(n, 2\ell)$ containing $R$ all fall in the sets $\mathcal{R}$ and $\Lcal$ respectively. For our purposes, the pseudo-randomness parameter $\eps$ of
$\Lcal$ should be thought of as close to $\mu(\Lcal)$, in which case the bound we get is achieved by randomly chosen $\Lcal, \Rcal$ of given densities.

\begin{lemma} \label{lm: pseudorandom edges}
For any $\delta > 0$, there is sufficiently large $\ell$ such that the following holds. Let $\Lcal \subseteq \Gras(n, 2\ell)$ and $\Rcal \subseteq \Gras(n, 2(1-\delta)\ell)$ be sets with fractional sizes $\alpha$ and $\beta$ respectively, and suppose that $\Lcal$ is $(r,\eps)$ pseudo-random.
Then for all $t\geq 4$ that are powers of $2$,
\[ \Pr_{R, L_1, \ldots, L_k \supseteq R}[R \in \Rcal, L_1, \ldots, L_k \in \Lcal] \leq 2^{O_{t,r,k}(1)} \beta^{(t-1)/t} \eps^{(kt-2)/t} + 2^{-r\delta\ell+2},
\]
where in the above probability $R \in \Grass(n, 2(1-\delta)\ell))$ is chosen uniformly, and each $L_i \in \Grass(n, 2\ell)$ is chosen independently and uniformly conditioned on $L_i\supseteq R$.
\end{lemma}
\begin{proof}
    Deferred to \cref{sec: proof of pseudorandom edges}.
\end{proof}

Using \cref{lm: pseudorandom edges} we are able to prove the soundness of the $(k+1)$-query Grassmann consistency test 
as stated below

\begin{thm}\label{th: consistent with side}
For any $\delta > 0$, there is sufficiently large $\ell$ such that the following holds for $r^\star = \frac{10k}{\delta}$. Let $U$ be a question to the first prover, let $T_1, T_2$ be tables as above, and suppose that 
\[
\Pr_{\substack{R: \; R \cap H_U = \{0\}\\ L_1, \ldots, L_k: \;  L_i \cap H_U = \{0\},  L_i\supseteq R}}[T_1[L_i\oplus H_U]|_R = T_2[R], \; \forall i \in [k]] = \epsilon \geq 2^{-2(1-1000\delta)\ell k},
\]
where $R \in \Grass(\Ff_2^U, 2(1-\delta)\ell)$ and each $L_i \in \Grass(\Ff_2^U, 2\ell)$ in the probability above. Then there are parameters $r_1$ and $r_2$ such that $r_1 + r_2 \leq r^\star$, such that for at least $2^{-6\ell^2}$ fraction of the $r_1$-dimensional subspaces $Q \subseteq \Ff_2^{U}$, there exists $W \subseteq \Ff_2^U$ of codimension $r_2$ containing $Q \oplus H_U$, and a global linear function $g_{Q,W}: W \xrightarrow[]{} \Ff_2$ that respects the side conditions on $H_U$ such that
    \[
    \Pr_{L}[g_{Q,W}|_{L\oplus H_U} = T_1[L\oplus H_U] \;|\; Q \subseteq L \subseteq W] \geq \frac{2^{-2(1-1000\delta^2)\ell}}{5}.
    \]
\end{thm}
In the remainder of the section, we first show how to obtain \cref{th: consistent with side} using \cref{lm: pseudorandom edges} and then prove \cref{lm: pseudorandom edges}.

\subsection{Proof of~\cref{th: consistent with side}}
We start by proving a simpler version of~\cref{th: consistent with side} without the side conditions $H_U$, 
and then explain how to deduce~\cref{th: consistent with side} from it.
Let $T_1$ be a table that assigns to each $L \in \Grass(\Ff_2^U, 2\ell)$ a linear function over $L$, and let $T_2$ be a table that assigns to each $R \in \Grass(\Ff_2^U, 2(1-\delta)\ell)$ a linear function over $R$. The test  chooses an entry $R$ of $T_2$, and then $k$ entries of $T_1$, $L_1, \ldots, L_k$ such that each $L_i \supseteq R$, and accepts if and only if $T_1[L_i]|_R \equiv T_2[R]$ for all $i \in [k]$. 
%After analyzing this test, it is straightforward to deduce a similar decoding result for the version with side conditions considered in \cref{th: consistent with side}. 

% Our analysis of the test of this subsection shows that if the pass probability at least $2^{-2(1-1000\delta)\cdot k}$, then there must be a zoom-in, zoom-out pair $(Q, W)$ of relatively small dimension such that the table $T_1$ has non-trivial agreement with a linear function inside of $\Zoom[Q, W]$. After we show \cref{th: consistency}, we use straightforward reductions from \cite{MZ24} to show a slightly stronger version wherein the decoding result holds for many zoom-ins $Q$, and also takes into account the side conditions. This latter result is what is ultimately used in our PCP's soundness analysis. 

\begin{thm} \label{th: consistency}
    Suppose that tables $T_1$ and $T_2$ satisfy
    \[
   \Pr_{R; L_1, \ldots, L_k \supseteq R}[T_1[L_i]|_R \equiv T_2[R] \; \forall i \in [k]] = \epsilon \geq 2^{-2(1-1000\delta)\ell k} .
    \]
     Then there exist subspaces $Q \subset W$ and a linear function $f: W \xrightarrow[]{} \Ff_2$ such that:
    \begin{enumerate}
    \item $\codim(Q) + \dim(W) = \frac{10k}{\delta} := r^\star$.
    \item $f|_L \equiv T_1[L]$  for $\Omega(\epsilon')$-fraction of $2\ell$-dimensional $L \in {\Zoom}[Q,W]$, where $\epsilon' := 2^{-2(1-1000\delta^2)\ell}$.
\end{enumerate}
\end{thm}
\begin{proof}
We use an argument of~\cite{BKS}. Choose a linear function $f: \Ff_2^n \xrightarrow[]{} \Ff_2$ uniformly at random and define the (random) sets of vertices
\[
S_{L,f} = \{L \in \Grass(n, 2\ell) \; | \; f|_L \equiv T_1[L] \} \quad \text{and} \quad S_{R,f} = \{R \in \Grass(n, 2(1-\delta)\ell)\; | \; f|_R = T_2[R] \}.
\]

Denote by $E_k(S_{L,f}, S_{R,f})$ to be the set of tuples $(R, L_1, \ldots, L_k)$ such that $R \subseteq L_i$, $L_i \in S_{L, f}$ for each $i$ and $R \in S_{R, f}$. Let $M$ be the total number of tuples $(R, L_1, \ldots, L_k)$ where $R \subseteq L_i$ for each $i \in [k]$ and $L_i\cap L_j = R$ for each $i\neq j$. Then
\[
\E_f \left[\left|E_k(S_{L,f}, S_{R,f})\right| \right] \geq \eps M\cdot 2^{-2(1-1000\delta)\ell} 2^{-2k\delta \ell},
\qquad \E_f[\mu(S_{R,f})] = 2^{-2\ell(1-1000\delta)},
\]
so by linearity of expectation
\[
\E_f \left[\left|E_k(S_{L,f}, S_{R,f})\right| -  \eps \mu(S_{R,f}) M2^{-2k\delta \ell} /2  \right] \geq \frac{\eps \cdot 2^{-2(1-1000\delta)\ell}M 2^{-2k\delta \ell}}{2} .  
\]
It follows that there exists $f$ achieving this expectation and we fix this $f$ henceforth, so that
\begin{equation} \label{eq: grassmann soundness 1} 
\left|E_k(S_{L,f}, S_{R,f})\right| \geq \frac{\eps \mu(S_{R,f})M2^{-2k\delta \ell}}{2} + \frac{\eps \cdot 2^{-2(1-\delta)\ell}M2^{-2k\delta \ell}}{2}.
\end{equation}

Now we let $S_{L,f}$ and $S_{R,f}$ be defined relative to the $f$ we fixed and we claim that $S_{L, f}$ is not $(r^\star, \eps')$ pseudo-random for $\eps' = 2^{-2\ell(1-1000\delta^2)}$. Let $\beta := \mu(S_{R,f})$. First note that $\beta$ cannot be too small. Indeed,
\[
\beta = \mu(S_{R,f}) \geq \frac{|E_k(S_{L,f}, S_{R,f})|}{M},
\]
so $\beta \geq \eps \cdot 2^{-2\ell-2k\delta\ell} \geq 2^{-2(k+1)\ell}$. Now, if $S_{L, f}$ were $(r^\star,\eps')$-pseudo-random, then~\cref{lm: pseudorandom edges} would imply 
\begin{equation} \label{eq: grassmann soundness 2}   
\begin{split}
\frac{\left|E_k(S_{L,f}, S_{R,f})\right|}{M} &
\leq 2\Pr_{R, L_1, \ldots, L_k \supseteq R}[R \in \Rcal, L_1, \ldots, L_k \in \Lcal] \\
&\leq 2^{O_{r^\star,t,k}(1)}\beta^{(t-1)/t} \eps'^{(kt-2)/t} + 2^{-r^\star\delta \ell+2} \\
&\leq  2^{O_{r^\star,t,k}(1)}\beta^{(t-1)/t} \eps'^{(kt-2)/t}.
\end{split}
\end{equation}
In the first transition we used the fact that for randomly chosen $R, L_1,\ldots,L_k$, $L_i\cap L_j = R$ for all $i\neq j$ with probability at least $1/2$, and in the last transition we use the fact that the second term on the right hand side of the second line is at most $2^{-10\ell k}$ 
and is negligible compared to the first term there. Now, combining \eqref{eq: grassmann soundness 1} and \eqref{eq: grassmann soundness 2}, we have,
\[
\eps \beta \leq 2^{1+2k\delta\ell} \cdot \frac{\left|E_k(S_{L,f}, S_{R,f})\right|}{M}  \leq 2^{O_{r^\star,t,k}(1)}2^{2k\delta\ell}\beta^{(t-1)/t} \eps'^{(kt-2)/t},
\]
or equivalently $1 \leq  2^{O_{r^\star,t,k}(1)} 2^{2k\delta\ell}\beta^{-1/t} \frac{\eps'^{(kt-2)/t}}{\eps}$.
Using $\beta\geq 2^{-2(k+1)\ell}$, $\eps \geq 2^{-2(1-1000\delta)\ell k}$, and $\eps' = 2^{-2(1-1000\delta^2)\ell}$, the 
last inequality implies
\begin{align*}
1 &\leq  2^{O_{r^\star,t,k}(1)} 2^{2k\delta\ell}\beta^{-1/t} \frac{\eps'^{(kt-2)/t}}{\eps}\\
&=  2^{O_{r^\star,t,k}(1)}2^{2k\delta\ell}\cdot \beta^{-1/t} \cdot \eps'^{-2/t} \cdot \frac{\eps'^k}{\eps} \\
&\leq 2^{O_{r^{*},t,k}(1)}2^{2k\delta\ell}
\left(2^{-2(k+1)\ell}\eps'^2\right)^{-\frac{1}{t}}2^{-2000(\delta-\delta^2)\ell k}
 \\
 &\leq 
2^{O_{r^{*},t,k}(1)}2^{2k\delta\ell}
2^{\frac{4(k+1)}{t}\ell}
2^{-2000(\delta-\delta^2)\ell k},
\end{align*}
so for 
$t = (k+1)/(\delta-\delta^2)$ we get
$1\leq 2^{O_{r^{*},t,k}(1)}2^{2k\delta\ell}2^{4\cdot (\delta-\delta^2)\ell -2000(\delta-\delta^2)\ell k}$.
This is a contradiction for sufficiently large $\ell$ relative to $r^\star$ and $k$, and for this setting of $\ell$ it follows that $S_{L, f}$ is not $(r^\star, \eps')$ pseudo-random.
\end{proof}
% Passing from \cref{th: consistency} above to the version with side conditions in \cref{th: consistent with side} is straightforward and done in \cite{MZ24}.

% From \cref{th: consistency} we can deduce the version of the decoding theorem that we ultimately need using results from \cite{MZ24}.

\begin{proof} [Proof of \cref{th: consistent with side}]
    The proof is the same as that of \cite[Theorem 5.3]{MZ24} except one refers to \cref{th: consistency} whenever Theorem 5.1 from \cite{MZ24} is used. One must also show a version of \cite[Theorem 5.2]{MZ24}, 
    %which starts with the assumption that $T_1$ and $T_2$ pass the $(k+1)$-query Grassmann consistency test with probability at least $\eps$, 
    and this too can be done by following the proof of \cite[Theorem 5.2]{MZ24} and replacing their Theorem 5.1 with \cref{th: consistency} above.
\end{proof}
\subsection{Proving \cref{lm: pseudorandom edges} by Passing to the Bilinear Scheme} \label{sec: proof of pseudorandom edges}
In this section we prove~\cref{lm: pseudorandom edges}. To do so, we will  express the probability of interest there in terms of an inner product of functions over the bilinear scheme, and then bound it using Fourier analysis. At a high level, we do the following steps:
\begin{enumerate}
    \item Define Boolean functions $F: \Ff_2^{n \times 2\ell} \to \{0,1\}$ and $G: \Ff_2^{n \times 2(1-\delta)\ell} \to \{0,1\}$ which act as indicator functions for $\Lcal$ and $\Rcal$ respectively.
    \item Write the probability from \cref{lm: pseudorandom edges} as an expression involving $F$ and $G$.
    \item Show that $\Lcal$ being pseudo-random implies that $F$ is pseudo-random.
    \item Bound the expression from step 2 using  global hypercontractivity.
\end{enumerate}

Fix $\Lcal$ and $\Rcal$ as in \cref{lm: pseudorandom edges} and towards the first step define the functions $F: \Ff_2^{n \times 2\ell} \to \{0,1\}$ and $G: \Ff_2^{n \times 2(1-\delta)\ell} \to \{0,1\}$ as:

  \begin{equation*}
    F(M')=
    \begin{cases}
      1  & \text{if } \im(M') \in \Lcal, \\
      0 & \text{otherwise},
    \end{cases}
    \qquad\qquad 
    G(M)=
    \begin{cases}
     1 & \text{if } \im(M) \in \Rcal, \\
      0 & \text{otherwise}.
    \end{cases}
  \end{equation*}
  In the above, $\im(\cdot)$ is the usual definition of matrix image and refers to the span of the columns of the matrix. Implicit in the definitions is the fact that $F$ evaluates to $0$ if the columns of $M'$ are not linearly independent, as in this case their span is not even a dimension $2\ell$ subspace, and likewise $G$ evaluates to $0$ if the columns of $M$ are not linearly independent. We also note that $F$ and $G$ both satisfy a property called \emph{basis invariance}, which is needed to apply results from \cite{MZ24}. A function $H \in L_2(\Ff_2^{n \times 2 \ell})$ is called basis invariant if $H(M) = H(MA)$ for any full rank $A \in \Ff_2^{2\ell \times 2\ell}$. 
  %It is clear that both $F$ and $G$ are basis invariant.

To write down the expression for the probability from \cref{lm: pseudorandom edges}, we use an adjacency operator $\T: L_2\left(\Ff_2^{n \times 2\ell}\right) \xrightarrow{} L_2\left(\Ff_2^{n \times 2(1-\delta)\ell} \right)$, which corresponds to a bipartite graph between $\Ff_2^{n \times 2\ell}$ to $\Ff_2^{n \times 2(1-\delta)\ell}$. This graph morally corresponds to the bipartite graph between $2\ell$-dimensional subspaces and $2(1-\delta)\ell$-dimensional subspaces whose edges are $(L,R)$ satisfying $L \supseteq R$. Formally, for any $H \in L_2(\Ff_2^{n \times 2\ell})$, the function 
$\T H \colon \Ff_2^{n \times 2(1-\delta)\ell}\to\mathbb{C}$ is given by
\[
\T F(M) = \E_{v_1,\ldots, v_{2\delta \ell}}[F\left([M, v_1,\ldots, v_{2\delta \ell}]\right)].
\]
Here, $[M, v_1,\ldots, v_{2\delta \ell}]$ refers to the matrix obtained by adding the columns $v_i$ to $M$ on the right, for all $i \in [2\delta\ell]$. In words, the operator $\T$ averages over extensions of the matrix
$M$ to an $n\times 2\ell$ matrix by adding to it $2\delta\ell$ random columns. To see how the induced edges of $\T$ relate to subspaces $R \subseteq L$, observe that for $M \in \Ff_2^{n \times 2(1-\delta)\ell}$ and $M' = [M, v_1,\ldots, v_{2\delta \ell}]$, we indeed have $\im(M) \subseteq \im(M')$. Using the adjacency operator $\T$, we can now relate the probability from \cref{lm: pseudorandom edges} to an inner product of functions over the bilinear scheme.

\begin{lemma} \label{lm: relate edges to inner product}
    For $\Lcal, \Rcal, F, G$ as above we have $
    \Pr_{R; L_1, \ldots, L_k \supseteq R}[R \in \Rcal, L_1, \ldots, L_k \in \Lcal] \leq 2 \inner{(\T F)^k}{G}$.
\end{lemma}
\begin{proof}
For simplicity, we let $p_{k}$ denote the probability on the left hand side in the lemma statement. We expand the inner product on the right hand side:
   \begin{align*}
   \inner{(\T F)^k}{G} = \E_{M}\left[G(M) \left(\E_{M'}[F(M')]\right)^k\right] = \E_{M, M'_1, \ldots, M'_k}[G(M) \cdot F(M'_1) \cdots F(M'_k)],
   \end{align*} 
   where in the above expectations $M$ is uniformly random in $\Ff_2^{n \times 2(1-\delta)\ell}$ and $M'$ as well as each $M'_i$ are uniformly random in $\Ff_2^{n \times 2\ell}$ conditioned on the first $2(1-\delta)\ell$ columns equaling $M$. Now let ${\sf A}$ be the event each one of $M, M'_1, \ldots, M'_k$ in the last expectation has linearly independent columns. Note that conditioned on ${\sf A}$, $\im(M)$ is a uniformly random subspace in $\Grass(n, 2(1-\delta)\ell)$, and each $\im(M'_i)$ is a uniformly random subspace in $\Grass(n, 2\ell)$ containing $\im(M)$. Unpacking the definitions of $G$ and $F$, we get that, 
   \[
   \E_{M, M'_1, \ldots, M'_k}[G(M) \cdot F(M'_1) \cdots F(M'_k) \; | \; {\sf A}] =  p_k,
   \]
   Using the fact that the quantity in the expectation always has magnitude at most $1$, we get
   \[
    \inner{(\T F)^k}{G} \geq \frac{1}{1-\Pr[{\sf A}]} \cdot p_k.
   \]
   To conclude, we use the following simple bound on $\Pr[{\sf A}]$:
   \[
   \Pr[{\sf A}] \leq (2\ell + 2\delta k \ell) \cdot \frac{2^{2\ell}}{2^n} \leq \frac{1}{2}.
   \]
  To obtain this bound note that choosing $M_1, M'_1, \ldots, M'_k$ can be thought of as choosing $2\ell + 2\delta \ell k$ vectors from $\Ff_2^{n}$, so we can union bound over the probability that any vector creates a linear dependence in its respective matrix. Putting everything together, we get $p_k \leq 2\inner{(\T F)^k}{G}$.
\end{proof}

With \cref{lm: relate edges to inner product} we have reduced showing \cref{lm: pseudorandom edges} to bounding the analytical quantity $\inner{(\T F)^k}{G}$. Moreover, we know that $F$ comes from a pseudo-random set of subspaces. Intuition from Fourier analysis over the bilinear scheme says that if the function $F$ is pseudo-random (in the sense of \cref{def: bilinear scheme pseudorandom}) then this quantity should behave as if $F$ is a randomly chosen indicator function. Using global hypercontractivity, we will show that this is indeed the case. 
\begin{lemma} \label{lm: preserve pseudorandom}
   The function $F$ is $(r, 2\epsilon)$-pseudo-random.
\end{lemma}
\begin{proof}
See \cite[Lemma A.18]{MZ24}.
\end{proof}
We now introduce the global hypercontractivity machinery, and we present a few tools from \cite{EvraKL} and \cite{MZ24}. We remark that the next three results hold for any function $F: \Ff_2^{n \times 2\ell} \to \{0,1\}$ that is pseudo-random, not just the one that we consider. 
The first theorem is the ``global hypercontractivity'' theorem over the bilinear scheme and bounds the norm of the level $d$ of a pseudo-random function.

\begin{thm}[\cite{EvraKL}] \label{th: EKL}
    Let $t \geq 4$ be a power of $2$. If $F: \Ff_2^{n \times 2\ell} \xrightarrow[]{} \{0,1\}$ is $(r, \epsilon)$-pseudo-random and $d\leq r$, then $\norm{F^{=d}}_t \leq 2^{500d^2 t}\epsilon^{\frac{t-2}{t}}$.
\end{thm}
\begin{proof}
The theorem is implicit in \cite{EvraKL}, and we refer to \cite[Theorem A.7]{MZ24} for this exact formulation.
\end{proof}
The next result states $\T$ decays the norm of the level $d$ part of basis-invariant $F$ by factor close to $2^{-2d\delta \ell}$. 
% This comes from the fact that the level $d$ portion, $F^{=d}$, is, up to a small error, an eigenfunction of $\T$ with singular value $2^{-d\delta \ell}$. 

\begin{lemma} \label{lm: not increase norm}
    If $F \in L_2\left(\Ff_2^{n \times 2 \ell}\right)$ is basis invariant, 
    then for any $0 \leq d \leq 2 \ell$, we have
    \[
    \norm{\T F^{=d}}^2_2 \leq \left(2^{-d (2\delta \ell-1)} + 3 \cdot 2^{d-n}\right)\norm{F^{=d}}^2_2.
    \]
\end{lemma}
\begin{proof}
In Appendix A of \cite{MZ24} an operator $\mc{G}: L_2\left( \Ff_2^{n \times 2(1-\delta)\ell}\right) \to  L_2\left( \Ff_2^{n \times 2\ell}\right)$ is defined along with an operator  $\Phi:   L_2\left( \Ff_2^{n \times 2\ell}\right) \to L_2\left( \Ff_2^{n \times 2\ell}\right)$, which is almost the adjacency operator of a Cayley Graph, and is given by 
\[
\Phi G (M) = \E_{\substack{B \in \Ff_2^{n \times 2\delta \ell}\\ C \in \Ff_2^{2\delta \ell \times 2\ell}\\ \rank(C) = 2\delta\ell}} \left[G(M + BC) \right].
\]
In Lemma A.12 of \cite{MZ24} it is shown that, since $F$ is basis invariant, we have 
\begin{equation} \label{eq: phi}   
\mc{G} \T F= \Phi F.
\end{equation}
By Lemma A.10 of~\cite{MZ24} we have that $F^{=d}$ is basis invariant, so by Lemma A.11 there we get
\begin{equation} \label{eq: pseudo adjoint}
 \norm{\T F^{=d}}^2_2 =   \inner{\T F^{=d}}{\T F^{=d}} = \inner{F^{=d}}{\mc{G} \T F^{=d}} = \inner{F^{=d}}{\Phi F^{=d}}.
\end{equation}
The first equality is by definition of $2$-norm and the last equality uses \eqref{eq: phi}.

Now we use the fact that $F^{=d}$ can be written as a linear combination of $\chi_S$ for $S \in \Ff_2^{n \times 2 \ell}$ of rank $d$, as described in \eqref{eq: level d decomp}. By Lemma A.13 of~\cite{MZ24}, each such $\chi_S$ for $S$ of rank $d$ is an eigenvector of $\Phi$ with eigenvalue at most $2^{-d (2\delta \ell-1)} + 3 \cdot 2^{d-n}$. Thus, for each $S$ of rank $d$
\[
\norm{\Phi \chi_S}_2 \leq \left(2^{-d (2\delta \ell-1)} + 3 \cdot 2^{d-n}\right) \cdot \norm{\chi_S}_2.
\]
Because the $\chi_S$ are orthogonal we conclude that
\begin{equation} \label{eq: eigenfunction}
   \norm{\Phi F^{=d}}_2 \leq \left(2^{-d (2\delta \ell-1)} + 3 \cdot 2^{d-n}\right) \cdot \norm{F^{=d}}_2. 
\end{equation}
To conclude, we combine \eqref{eq: eigenfunction} with \eqref{eq: pseudo adjoint} and the Cauchy-Schwarz inequality 
\end{proof}

We now state and prove a version of~\cref{lm: pseudorandom edges} for the bilinear scheme:
\begin{lemma} \label{lm: pseudorandom edge count}
    Let $F \in L_2\left( \Ff_2^{n \times 2 \ell} \right)$ and $G \in L_2\left(\Ff_2^{n \times 2(1-\delta)\ell} \right)$ be basis invariant indicator functions where $F$ is $(r, \eps)$ pseudo-random and $G$ has expectation $\E[G] = \beta$. Then for all $t \geq 4$ that are powers of $2$, we have 
    \[
    \inner{(\T F)^k}{G} \leq  2^{O_{r,t,k}(1)}  \beta^{(t-1)/t}\eps^{(kt-2)/t} + 2^{-r\delta \ell+1}.
    \]
\end{lemma}
\begin{proof}
    Using the degree decomposition, we can write $\T F = \sum_{i = 0}^{2\ell} \T F^{=i}$. We partition the terms of the summation into $L = \sum_{i = 0}^{r} \T F^{=i}$ and $H = \sum_{i = r+1}^{2\ell} \T F^{=i}$, so that altogether $L + H = \T F$. We will want to start by bounding the probability that $H$ is large, and to this end we first calculate its norm. Note that for $i> j$ 
    we have that
    \begin{equation}\label{eq:to_parsveal}
    \inner{\T F^{=i}}{\T F^{=j}}
    =\inner{F^{=i}}{\mathcal{G}\T F^{=j}}
    =\inner{F^{=i}}{\Phi F^{=j}}
    =0.
    \end{equation}
    The first transition is by combining Lemmas A.10 and A.11 of~\cite{MZ24}, 
    the second transition is by Lemma A.12 of~\cite{MZ24}, and the last lemma is because by Lemma A.13 of~\cite{MZ24}, the operator $\Phi$
    preserves the space of degree at most $j$ functions. Combining with~\Cref{lm: not increase norm} it follows that
    \[
    \E_{M}[|H(M)|^2] = \norm{H}_2^2 = \sum_{i = r+1}^{2\ell} \norm{\T F^{=i}}_2^2 \leq \left(2^{-r(2\delta \ell-1)} + 3 \cdot 2^{r-n} \right) \sum_{i=r+1}^{2\ell} \norm{F^{=d}}_2^2 \leq 2^{-r(2\delta\ell-1)}.
    \]
    Letting $E$ be the event that $|H(M)| \leq 2^{-0.9r\delta \ell}$, we get by Markov's inequality that $\Pr[\bar{E}] \leq 2^{-r\delta\ell}$. We have:
     \begin{align*}     
     \inner{(\T F)^k}{G}  &= \E_{M}[(\T F)(M)^k \cdot G(M) \cdot \ind(E)] +  \E_{M}[(\T F)(M)^k \cdot G(M)  \cdot \ind(\bar{E})] \\
     &\leq \E_{M}[G(M) \cdot (L(M) + H(M))^k \cdot \ind(E)] + \Pr[\bar{E}] \\
     &\leq 2^k\left(\E_{M}[G(M) \cdot |L(M)|^k \cdot \ind(E)] +   \E_{M}[G(M) \cdot |H(M)|^k \cdot \ind(E)] \right) + 2^{-r\delta\ell}.
     \end{align*}
     In the second transition we used the fact that $\T F$ and $G$ are $1$-bounded, and in the third transition we used $|a+b|^k\leq 2^k(|a|^k+|b|^k)$.
     Next, the second term on the right hand side is upper bounded as:
     \[
     \E_{M}[G(M) \cdot |H(M)|^k \cdot \ind(E)] \leq 2^{-0.9r\delta \ell k } \E_{M}[G(M)] \leq 2^{-0.9r\delta \ell k }. 
     \]
     For the first term above, we first use H\"{o}lder's inequality  to get
     \begin{align*}
         \E_{M}[G(M) \cdot |L(M)|^k \cdot \ind(E)] &\leq  \E_{M}[G(M) \cdot |L(M)|^k ]  
         \leq \norm{G}_{t/(t-1)} \norm{L^k}_{t} 
         = \norm{G}_{t/(t-1)} \norm{L}^k_{kt}.
     \end{align*}
     By Booleanity of $G$, 
     $\norm{G}_{t/(t-1)} = \beta^{(t-1)/t}$. To bound $\norm{L}^k_{kt}$ we use the triangle inequality along with the fact that $\norm{\T F^{=i}}_{kt}\leq \norm{F^{=i}}_{kt}$ by Jensen's inequality and~\cref{th: EKL} as follows:
     \begin{align*}
    \norm{L}^k_{kt} &\leq \left(\sum_{i=0}^r \norm{ \T F^{=i}}_{kt} \right)^k 
    \leq  \left(\sum_{i=0}^{r} \norm{F^{=i}}_{kt} \right)^k  
    \leq  \left( 2^{O_{r,t,k}(1)} \eps^{\frac{kt-2}{kt}} \right)^k 
    \leq 2^{O_{r,t,k}(1)} \eps^{(kt-2)/t}.
         \end{align*}
    Putting everything together gives
    \[
      \inner{(\T F)^k}{G} \leq \beta^{(t-1)/t} 2^{O_{r,t,k}(1)} \eps^{(kt-2)/t} 
      + 2^k \cdot 2^{-0.9r\delta \ell k } + 2^{-r\delta \ell },
    \]
    where the second term is negligible compared to the third term.
\end{proof}

Finally, we can put all of the pieces together and show \cref{lm: pseudorandom edges}.
\begin{proof}[Proof of~\cref{lm: pseudorandom edges}]
Let $F$ and $G$ be as above. We have
\[
\norm{F}_2^2 \leq \mu(\Lcal) = \alpha, \quad \norm{G}_2^2 \leq \mu(\Rcal) = \beta.
\]
By \cref{lm: relate edges to inner product} we have
$\Pr_{R; L_1, \ldots, L_k \supseteq R}[R \in \Rcal, L_1, \ldots, L_k \in \Lcal] \leq  2\langle \T F, G \rangle$,
and by \cref{lm: preserve pseudorandom} $F$ is $(r, 2\epsilon)$-pseudo-random. Now, applying \cref{lm: pseudorandom edge count} we conclude
\begin{align*}
\Pr_{R; L_1, \ldots, L_k \supseteq R}[R \in \Rcal, L_1, \ldots, L_k \in \Lcal] \leq 2\langle \left(\T F\right)^k, G \rangle\leq 2^{O_{r,t,k}(1)} \beta^{(t-1)/t} (2\eps)^{(kt-2)/t} + 2^{-r\delta\ell+2},
\end{align*}
which is at most $2^{O_{r,t,k}(1)} \beta^{(t-1)/t} \eps^{(kt-2)/t} + 2^{-r\delta\ell+2}$.
\end{proof}

\section{Analysis of the PCP} 
In this section, we analyze the composed PCP and prove \cref{thm: main}. 

\subsection{Tools for the Soundness Analysis}
\label{sec:tools_for_pcp_analysis} \label{sec: covering properties} \label{sec: bounded zoom-outs statement}
To start, we present various tools from \cite{MZ24} that are needed for our soundness analysis.
\paragraph{The Covering Property} 
The covering properties we use are the same as those of \cite{KS, KMS}. Set the following parameters:
\begin{equation} \label{eq: pcp parameters}
    J = 2^{100\ell^2} \quad , \quad \beta = \frac{\log \log J}{J}.
\end{equation}
%where $c > 0$ is a constant arbitrarily small relative to $\delta$. 
We also set $\ell$ sufficiently large relative to $r^\star$ throughout the section. 

Fix a question $U = (e_1, \ldots, e_J)$ to the first prover and recall that $H_U = \spa(x_{e_1}, \ldots, x_{e_J})$ where $x_{e_i}$ is the vector that is one at coordinates corresponding to variables in $e_i$ and $0$ elsewhere. Let $Q \subseteq \Ff_2^U$ be a subspace of dimension $r$, for some $r \leq r^\star$. We will need covering properties to relate the following pairs of distributions. The first pair of distributions is for the advice subspace, while the second pair of distributions are for $2\ell$-dimensional subspaces conditioned on containing some advice subspace $Q$.

\paragraph{$\D_{r_1}$ and $\D'_{r_1}$ for some integer $r_1$.}
\begin{itemize}
    \item $\mc{D}_{r_1}:$ Choose a uniformly random $r_1$-dimensional subspace $Q \subseteq \Ff_2^U$.
    \item $\mc{D}'_{r_1}:$ Choose $V \subseteq U$ according to the outer PCP. Then, choose a uniformly random $r_1$-dimensional subspace $Q \subseteq \Ff_2^V$ and consider $Q$ as a subspace of $\Ff_2^U$ by inserting $0$'s in the missing coordinates.
\end{itemize}
\paragraph{$\D_{Q,W}$ and $\D'_{Q,W}$ for subspaces $Q \subseteq W$.}
\begin{itemize}
    \item $\mc{D}_{Q,W}:$ Choose a uniformly random $2\ell$-dimensional subspace $L \subseteq \Ff_2^U$ conditioned on $Q \subseteq L \subseteq W$.
    \item $\mc{D}'_{Q,W}:$ Sample a $2\ell$-space $L$ according to the following distribution conditioned on $Q \subseteq L \subseteq W$. Choose $V \subseteq U$ according to the outer PCP. Then, choose a uniformly random $L \subseteq \Ff_2^V$ conditioned on $L \supseteq Q$, and consider $L$ as a subspace of $\Ff_2^U$ by inserting $0$'s in the missing coordinates.  
\end{itemize}

In words, $\mc{D}_Q$ outputs the span of $2\ell$ vectors uniformly from $\Ff_2^U$ conditioned on the first $r$ of them spanning $Q$. The distribution $\mc{D}'_Q$ also outputs the span of $2\ell$ vectors from $\Ff_2^U$ conditioned on the first $r$ spanning $Q$, but these vectors are obtained first from $\Ff_2^V$ uniformly, and then lifted to $\Ff_2^U$. 
% The covering property asserts that for sufficiently large sets, the difference in measures of these two distributions is negligible. 
% Set $\eta = 2^{-100\ell^{100}}$ to be a negligible value. 
The following two lemmas bound the \emph{statistical distance} between the pairs of distributions above. 
%We define statistical distance formally below. 
% They are in the spirit of the covering properties of \cite{KS, KMS2}, and we refer to statements in \cite{dor_thesis} for reference. 

\begin{definition}
    Given two discrete distributions $\D, \D'$ over a set $\mc{A}$, their statistical distance is
    \[
    {\sf SD}(\D, \D') := \max_{A \subseteq \mc{A}} |\D(A) - \D'(A)|. 
    \]
\end{definition}

\begin{lemma}\cite[Lemma 3.6]{dor_thesis} \label{lm: covering advice}
    For any dimension $r_1 \leq \frac{10}{\delta}$ we have ${\sf SD}(\D_{r_1}, \D'_{r_1}) \leq \beta \sqrt{J} \cdot 2^{r_1 + 4}$.
\end{lemma}

\begin{lemma}\cite[Lemma 3.8]{dor_thesis}\label{lm: covering}
For all dimensions $0 \leq r \leq 2\ell-1$, and $r' \in \mathbb{N}$ the following holds. For at least $\left(1 - \sqrt{\beta} J^{1/4}\right)$-fraction of $r$-dimensional subspaces $Q \subseteq \Ff_2^U$, we have
    \[
    {\sf SD}(\mc{D}_{Q,W}, \mc{D}'_{Q, W}) \leq \sqrt{\beta} J^{1/4} \cdot 2^{2\ell + 5} \cdot 2^{r'(2\ell-r)+5},
    \]
    for all subspaces $W \subseteq \Ff_2^U$ of codimension $r'$ containing $Q$.
\end{lemma}

% \begin{lemma} \label{lm: smooth Q}
%     For any $\Lcal \subseteq (\Ff_2^{U})^{2\ell}$, we have
%     \[
%     \Pr_{Q}\left[\D'_Q(\Lcal) \geq 0.8 \cdot \D_Q(\Lcal) - \eta^{20} \right] \geq 1 - 2\eta^{20},
%     \]
%     where $Q$ in the probability above is the span of $r_1$ uniformly random vectors in $\Ff_2^U$.
% \end{lemma}
% \begin{proof}
%     See~\cite[Lemma 5.5]{MZ24}.
% \end{proof}
%  Finally, we also need the distributions $\mc{D}_r$ and $\mc{D}'_r$, for some parameter $r$. These distributions are simply the distributions of the first $r$ vectors, $x_1, \ldots, x_r$ of $\mc{D}_Q$ and $\mc{D'}_Q$ respectively in the case $Q = \{0\}$. In other words, $\mc{D}_r$ is obtained by sampling $(x_1, \ldots, x_r) \in \Ff_2^{U \times r}$ uniformly at random, while $\mc{D}'_r$ is sampled by lifting $r$ vectors from $\Ff_2^V$ to $\Ff_2^U$ and adding $w_1, \ldots, w_r$ as in $\mc{D}'_Q$, and finally ignoring the conditioning in the last step. The point of these distributions is that they are how the advice is sampled, and the next lemma asserts that, for our purposes, they are the same up to negligible constant factors.

% \begin{lemma} \label{lm: covering zoom-in} 
% Let $\mathcal{Q}$ be a set of $r_1$-dimensional subspaces in $\Ff_2^U$ satisfying $\D_{r_1}(\mathcal{Q}) \geq 2^{-10\ell^{10}}$. Then, 
% \[
% \D'_{r_1}(\mathcal{Q}) \geq 0.8 \cdot \D_{r_1}(\mc{Q}).
% \]
% \end{lemma}
% \begin{proof}
% See~\cite[Lemma 5.6]{MZ24}
% \end{proof}
\paragraph{A Bound on the Number of Decoded Functions} 
In \cref{th: consistent with side}, we showed a decoding theorem that says if the inner test passes with sufficient probability, then for a significant fraction of low dimension zoom-ins $Q$, there is a low codimension zoom-out $W$ such that $T_1$ agrees non-trivially with a linear function inside $\Zoom[Q, W]$. Roughly speaking, showing soundness for the composed PCP requires the provers to agree on the same linear function with non-trivial property, and this in turn requires them to agree on the same zoom-in $Q$ and same zoom-out $W$.  
Since the decoding theorem holds for a significant fraction of low dimension zoom-ins $Q$, the provers can agree on a shared zoom-in with reasonable probability using their advice. The same cannot be said for zoom-outs however, and to circumvent this issue we need another tool. 

In this section, we present the notion of \emph{maximal} zoom-outs and the result from \cite{MZ24} which shows that the number of maximal zoom-outs relative to any fixed zoom-in is bounded.

\begin{definition} \label{def: maximal}
    Given a table $T$ on $\Grass(\Ff_2^V, 2\ell)$, a subspace $Q \subseteq \Ff_2^V$, a zoom-out $W \supseteq Q$, and a linear function $f: W \to \Ff_2$, we say that the zoom-out, function pair $(W, g_W)$ is $(C, s)$-maximal with respect to $T$ on $Q$ if the following holds. 
    \begin{itemize}
        \item 
        %$g|_W$ agrees with at least $C$-fraction of $T$'s entries in $\Zoom[Q, W]$,
        $\Pr_{L \in \Zoom[Q,W]}[g_W|_L \equiv T[L] ] \geq C$.
    \item  There does not exist another zoom-out function pair, $(W', g_{W'})$ such that $\Ff_2^V \supseteq W' \supsetneq W$, $g_{W'}: W' \xrightarrow[]{} \Ff_2$, $g_{W'}|_{W}\equiv g|_{W}$ and
    $\Pr_{L \in \Zoom[Q,W']}[g_{W'}|_L \equiv T[L]] \geq sC$.
    \end{itemize}
    In the case that $Q = \{0\}$, we simply say that $(W, g_W)$ is $(C, s)$-maximal with respect to $T$. We say that the codimension of the zoom-out, function pair $(W, g_W)$ is $\codim(W)$.
\end{definition}
In the above statement, $C$ should be thought of as small and 
$s$ should be thought of as an absolute constant. 
With this in mind, a zoom-out $W$ and a linear function on it 
$g_W$ is called maximal if there is no zoom-out $W'$
that strictly contains $W$, and an extension of $g_W$ 
to $g_{W'}$, that has the same agreement
with $T$ as $g_W$ (up to constant factors). Intuitively, a zoom-out, function pair is ``maximal'' if its agreement with the table $T$ is not explained by any larger zoom-out function pair. In \cite{MZ24}, it is shown that the number of maximal zoom-out function pairs is bounded with respect to every zoom-out. 

% To present this lemma, set the following constants, which should all be considered constant:
% \[
% \xi = \delta^5, \quad \delta_2 = \xi/100, \quad t = \left(2^{2+10/\delta_2}\right)!.
% \]

% \begin{lemma} \label{lm: zoom out contained in maximal}
%    Let $T$ be a table on $\Grass(\Ff_2^V, 2\ell)$, $Q \subseteq \Ff_2^V$, and $W \subseteq \Ff_2^V$ be a subspace of codimension $r$ containing $Q$. Suppose that there exists a linear function $g_{W}: W \xrightarrow[]{} \Ff_2$ such that 
%     \[
%     \Pr_{L \in \Grass(\Ff_2^V, 2\ell)}[g_{W}|_L \equiv T[L] \; | \; Q \subseteq L \subseteq W] \geq C.
%     \]
% Then there exists a subspace $W' \supseteq W$ and a linear function $g_{W'}: W' \xrightarrow[]{} \Ff_2$ such that $g_{W'}|_{W} \equiv g_W$ and $(g_{W'}, W')$ is $(C s^{-r}, s)$-maximal and a linear function $g_{W'}: W' \xrightarrow[]{} \Ff_2$ such that 
%     \[
%     \Pr_{L \in \Grass(\Ff_2^V, 2\ell)}[g_{W'}|_L \equiv T[L] \; | \; Q \subseteq L \subseteq W] \geq C.
%     \]
% \end{lemma}
% \begin{proof}
%  See Lemma 5.18 of \cite{MZ24}.
% \end{proof}
\begin{thm} \label{th: bounded zoom-out with zoom-in}
For any $\delta > 0$ there is sufficiently large $\ell$ such that the following holds.  For any table $T$ on $\Grass(\Ff_2^V, 2\ell)$ satisfying $|V| \geq 2^{\ell}$, any subspace $Q \subseteq \Ff_2^V$ of dimension $r_1 \leq 10J/\delta$ and any $C \geq 2^{-2(1-\delta^3)\ell}$, the number of $(C,\frac{1}{5})$-maximal zoom-out, function pairs of codimension at most $10k/\delta$ with respect to $T$ on $Q$ is at most $C^{-2} \cdot  2^{O_{k,\delta}(\ell )}$.
\end{thm}
\begin{proof}
We apply \cite[Theorem 5.26]{MZ24} with the parameter $\delta$ there set to $\delta/k$.
\end{proof} 

\paragraph{An Auxiliary Lemma} The last tool we present is an auxiliary lemma from \cite{KMS}. We restate the version from \cite{dor_thesis} below.
\begin{lemma} \label{lm: retain codim}
    Let $U$ be a fixed question to the first prover in the Outer PCP consisting of $3J$-variables in some set of $J$ equations. Let $V \subseteq U$ be a random question to the second prover chosen according to the Outer PCP. Let $W \subseteq \Ff_2^U$ be a subspace of codimension $s$. Then, with probability at least $1 - 2^{s+3} \beta^2 J$ over the choice of the question $V$, we have
    $\dim(W \cap \Ff_2^V) = |V|-s$.
\end{lemma}
\begin{proof}
See~\cite[Lemma 3.16]{dor_thesis}.
\end{proof}
\subsection{Proof of \cref{thm: main}}
We are now ready to analyze our $(k+1)$-query PCP, and show the completeness and soundness claims of \cref{thm: main}. We remind the reader that our final composed PCP, $\Psi_{k+1}$, is described in~\cref{sec:pcp_construct}. 
\subsection{Completeness}
Suppose that the $\Lin$ instance $(X, \Eq)$ has an assignment $\sigma: X \xrightarrow[]{} \Ff_2$ that satisfies at least $1-\epsilon_1$ of the equations in $\Eq$. Let $\Us \subseteq \U$ be the set of all $U = (e_1, \ldots, e_J)$ where all $J$ equations $e_1, \ldots, e_J$ are satisfied. Then, $|\Us| \geq (1- J\epsilon_1) \U$. We identify $\sigma$ with the linear function from $\Ff_2^{X}\to\Ff_2$, assigning the value $\sigma(i)$ to the $i$th elementary basis element $e_i$. Abusing notation, we denote this linear map by $\sigma$ 
as well.

For each $U \in \Us$ and vertex $L \oplus H_U$, we set $T_1[L \oplus H_U] \equiv \sigma|_{L \oplus H_U}$. Since $U \in \Us$, these assignments satisfy the side conditions. For all other $U$'s, set $T_1[L \oplus H_U]$ so that the side conditions of $H_U$ are satisfied and $T_1[L \oplus H_U]|_{L} \equiv \sigma|_L$. Such an assignment is possible because $L \cap H_U = \{0\}$. Similarly,
the table $T_2$ is defined as $T_2[R] \equiv \sigma|_{R}$.

Sampling a constraint, note that the constraint is satisfied whenever the $L' \oplus H_{U'}$ chosen in step $3$ of the test satisfies that $U' \in \Us$. As the marginal distribution of $L' \oplus H_{U'}$ is uniform,\footnote{This is true because first a clique is chosen with probability that is proportional to its size and then a vertex is sampled uniformly from the clique.} the distribution of $U'$ is uniform. It follows that the constraint
is satisfied whenever $U'\in \mathcal{U}_{{\sf sat}}$, which 
happens with probability at least $1-J\eps_1$. Thus, 
${\sf val}(\Psi_{k+1})\geq 1-J\eps_1$.
\subsection{Soundness} \label{sec: PCP soundness proof}
Here, we relate the soundness of the composed PCP to that of the outer PCP. More precisely, we show:
\begin{lemma} \label{lm: soundness outer to inner}
 For all $\delta>0$ there exists a constants $r\in\mathbb{N}$ and $c(\delta, k) > 0$
 such that the following holds.
 Let $G_{\beta, r}^{\otimes J}$ be the parallel repetition of the Smooth Variable versus Equation Game with advice described in~\cref{sec: final outer pcp}, and let $\Psi_{k+1}$ be the $(k+1)$-CSP described in~\cref{sec:pcp_construct} constructed using sufficiently large $\ell$. If $\val(G_{\beta, r}^{\otimes J}) < 2^{-c(\delta, k) \cdot \ell^2}$, then $\val(\Psi_{k+1}) \leq 2^{-2(1-1000\delta)\ell k}$. 
 %where $c(\delta)$ is some constant depending only on the value $\delta$ used in defining $\Psi_{k+1}$.
\end{lemma}

We fix $r$ as in \cref{lm: soundness outer to inner} for the remainder of the section. With the completeness already established, \cref{lm: soundness outer to inner} finishes the proof of \cref{thm: main} by showing the soundness claim. The rest of the section is devoted to the proof of~\cref{lm: soundness outer to inner} and makes heavy use of the tools from~\cref{sec:tools_for_pcp_analysis}. 

Towards the proof, assume that $\val(G_{\beta,r}^{\otimes J}) < 2^{- c(\delta, k) \cdot \ell^2}$, and suppose for the sake of contradiction that 
there are tables $T_1$ and $T_2$ that are $\epsilon$-consistent for $\epsilon \geq 2^{-2\ell(1-1000\delta)}$. To arrive at a contradiction, we will show that this implies strategies in the outer game, $G_{\beta,r}^{\otimes J}$, for the two provers that with success probability greater than $2^{- c(\delta, k) \cdot \ell^2}$. 
%In more detail, the two provers will use the $\epsilon$-consistent tables $T_1$ and $T_2$ to derive strategies which pass $G_{\beta,r}^{\otimes J}$ with probability at least $2^{-c(\delta)\cdot O(\ell)^2}$. Since this contradicts our starting assumption on the value of $G_{\beta,r}^{\otimes J}$, it must be the case that $T_1$ and $T_2$ cannot be $\epsilon$-consistent, and \cref{lm: soundness outer to inner} is shown.

\subsubsection{Clique Consistency}
To start, we first reduce to the case where $T_1$ satisfies a condition called \emph{clique-consistency}, as this property will be convenient to work with.

\begin{definition}
We say an assignment $T$ to $\mathcal{A}$ is clique-consistent 
if for every vertex $L_1 \oplus H_{U_1}$ and for every 
$L_2\oplus H_{U_2},L_3\oplus H_{U_3}\in [L_1\oplus H_{U_1}]$, 
the assignments $T[L_2\oplus H_{U_2}]$ and $T[L_3\oplus H_{U_3}]$
satisfy the $1$-to-$1$ constraint between $L_2\oplus H_{U_2}$
and $L_3\oplus H_{U_3}$ as specified in~\cref{lm: clique extension}.
\end{definition}

The following lemma shows that if 
$T_1$ and $T_2$ are $\eps$-consistent assignments to 
$\Psi_{k+1}$, then there is a clique-consistent assignment 
$\widetilde{T}_1$ such that $\widetilde{T}_1$ and $T_2$ are $\eps'$-consistent for $\eps'$ which is only negligibly smaller than $\eps$.
\begin{lemma}\label{lem:assume_clique_consistent}
    Suppose that the assignments $T_1$ and $T_2 $ are $\epsilon$-consistent, then there is a clique-consistent assignment $\widetilde{T}_1$ such that $\widetilde{T}_1$ and $T_2$ are $\left(\epsilon - 2^{-J}\right)$-consistent. 
\end{lemma}
\begin{proof}
    Partition $\mathcal{A}$ into cliques, $\mathcal{A} = \textsf{Clique}_1 \sqcup \cdots \sqcup \textsf{Clique}_{m}$. For each $i$, choose a random $L \oplus H_U \in \Cl_i$ uniformly, and for every $L' \oplus H_{U'} \in \Cl_i$ assign $\widetilde{T}_1[L' \oplus H_{U'}]$ in the unique way that is consistent with $T_1[L \oplus H_U]$ and the side conditions of $U'$ as described in~\cref{lm: clique extension}. It is clear that $\widetilde{T}_1$ is clique-consistent, and we next analyze the expected fraction
    of constraints that $\widetilde{T}_1$ and $T_2$ satisfy.
    
      Let $E$ be the event that $L_1 \oplus H_U, \ldots, L_k \oplus H_U$ are all in different cliques when a constraint is sampled. Note that conditioned on $E$, we have that the pairs of tables $T_1, T_2$ and $\widetilde{T}_1, T_2$ have the same pass probability. To see this, we can think of $\widetilde{T}_1$ as being defined in a manner coupled with the sampling of a constraint, conditioned on $E$:
      \begin{itemize}
          \item Choose $U \in \mathcal{U}$ uniformly at random, and $R, L_1, \ldots, L_k \subseteq \Ff_2^U$ according to the constraint sampling procedure in \cref{sec: constraint graph}, conditioned on $E$. Say $L_i \oplus H_U$ is in the clique $\mathcal{C}_i$,
          \item For each $i \in [k]$, choose $L'_i \oplus H_{U'_i} \in \mathcal{C}_i$ uniformly at random,
          \item For each $i \in [k]$ assign all entries in $\widetilde{T}_1$ from the clique $\mathcal{C}_i$ according to $T_1[L'_i \oplus H_{U'_i}]$ as described above.
          \item Assign the remaining entries of $\widetilde{T}_1$ according to the original process described above.
      \end{itemize}
      Thus, letting $\eps'$ be the fraction of constraints that $\widetilde{T}_1, T_2$ satisfy, we get
      $\eps' \geq \eps - \Pr[E]$.
      To conclude, we show that the probability of the event $E$ is negligible. Suppose that $R$ is chosen. Then for any distinct $i, j \in [k]$ we have that the probability of $L_i \oplus H_U$ and $L_j \oplus H_U$ being in the same clique is 
      \[
      \Pr_{L_i, L_j}[L_i \oplus H_U \in [L_j \oplus H_U]] \leq \frac{\qbin{J+2\ell}{2\ell}}{\qbin{3J}{2\ell}} \leq 2^{-2J}.
      \]
    Union bounding over at most $\binom{k}{2}$ pairs of $i, j$, we get $\Pr[E] \leq k^2 2^{-2J} \leq 2^{-J}$.
\end{proof}
Applying~\cref{lem:assume_clique_consistent} we conclude that there are clique-consistent assignments to $\Psi_{k+1}$ that are $(\eps-2^{-J})$-consistent, and henceforth we assume that $T_1$ is clique-consistent to begin with. Also, for the sake of simplicity, we will ignore the $-2^{-J}$ term as it is negligible, and assume $T_1$ is $\eps$-consistent.

We remark that, in the notation of~\cref{sec: constraint graph}, the benefit of having a 
clique-consistent assignment is that the 
constraint that the verifier checks is equivalent to checking that $T_1[L_i\oplus H_U]|_{R} \equiv T_2[R]$ for $k$ subspaces $L_1, \ldots, L_k$. These checks correspond to a $(k+1)$-query test performed within
the space $\Ff_2^U$, making the soundness analysis of the Grassmann consistency test in \cref{th: consistent with side} applicable. 

\subsubsection{A Strategy for the First Prover}
Let $p(U)$ be the consistency of $T_1$ and $T_2$ conditioned on $U$ being the question to the first prover and 
let $\Ug = \{U\in\U~|~p(U)\geq \eps/2\}$.
As $E_U[p(U)] \geq \epsilon$, 
by an averaging argument we have
\begin{equation} \label{eq: good U}
    \Pr_{U}[U \in \Ug] \geq \frac{\eps}{2}.
\end{equation}

Let $U \in \U$ be the question to the first prover and let $Q$ be the advice. If $U \notin \Ug$, then the first prover gives up, so henceforth assume that $U \in \Ug$. For such $U$, the test of the inner PCP passes with probability at least $\frac{\epsilon}{2}$. Formally, this means that for good $U$,
\begin{equation} \label{eq: inner grassmann pass}   
\Pr_{\substack{R : \; R \cap H_U = \{0\} \\ L_1, \ldots, L_k: \; L_i \cap H_U = \{0\}, L_i \supseteq R}}[T_1[L_i \oplus H_U]|_R \equiv T_2[R], \; \forall i \in [k]] \geq\frac{\epsilon}{2}.
\end{equation}

In the probability above, $R \in \Grass(\Ff_2^U, 2(1-\delta)\ell)$ and each $L_i \in \Grass(\Ff_2^U, 2\ell)$.
Next, the first prover chooses an integer $0 \leq r_1 \leq r$ uniformly, and takes $Q$ to be the span of the first $r_1$ advice vectors. Note that~\eqref{eq: inner grassmann pass} is exactly the condition of our decoding theorem for the Grassmann consistency test, \cref{th: consistent with side}. Thus applying~\cref{th: consistent with side}, we get that there are $r_1, r_2$ satisfying $r_1 + r_2 \leq r$ such that for at least $2^{-6\ell^2}$ of the $Q \subseteq \Ff_2^U$ of dimension $r_1$, there exists $W_Q \subseteq \Ff_2^U$, of codimension $r_2$, containing $Q \oplus H_U$  along with a linear function $g_{Q,W_Q}: W_Q \xrightarrow[]{} \Ff_2$ that satisfies
\begin{equation} \label{eq: p1 consistency}
\Pr_{L \subseteq \Ff_2^U}[g_{Q,W_Q}|_{L \oplus H_U} \equiv T_1[L \oplus H_U] \; | \; Q \subseteq L \subseteq W_Q]  \geq \frac{2^{-2(1-1000\delta^2)\ell}}{5} := C.
\end{equation}
We fix this particular pair of dimension and codimension $r_1,r_2$ henceforth. With probability at least $\frac{1}{r+1}$, the first prover chooses this dimension $r_1$, such that $Q, W_Q,$ and $g_{Q, W_Q}$ that satisfy \eqref{eq: p1 consistency} exist. Call an $r_1$-dimensional subspace $Q$ lucky if there are $W_Q \supseteq Q \oplus H_U$ and $g_{Q,W_Q}$ satisfying \eqref{eq: p1 consistency} and let $\Ql$ be the set of all lucky $Q \subseteq \Ff_2^U$. 

For each $Q$ such that $Q \in \Ql$ and $Q \cap H_U = \{0\}$, the first prover chooses a $W_Q$ of codimension at most $r$ and a linear function $g_{Q,W_Q}: W_Q \xrightarrow[]{} \Ff_2$ that satisfies the side conditions on $H_U$ and such that the inequality in~\eqref{eq: p1 consistency} holds. The first prover then extends $g_{Q,W_Q}$ to 
a linear function $y\rightarrow \inner{s_{Q,U}}{y}$ on $\mathbb{F}_2^U$ randomly, and outputs $s_{Q,U}$ as their answer.

 \subsubsection{A Strategy for the Second Prover}
Let $V$ be the question to the second prover. The second prover will use a table $\widetilde{T}_1$ to derive their strategy. The table $\widetilde{T}_1$ is obtained from $T_1$ as follows. For a question $V$ to the second prover, let $U \supseteq V$ be an arbitrary question to the first prover. For all $2\ell$-dimensional subspaces $L \subseteq \Ff_2^V$, define
\[
\widetilde{T}_1[L] \equiv T_1[L \oplus H_U]|_L.
\]
In order to make sure that $\widetilde{T}_1$ is well defined, we note two things. First, the subspace $L \oplus H_U$ can be viewed as a subspace of $\Ff_2^U$ because each $L \subseteq \Ff_2^{V}$ can be ``lifted'' to a subspace of $\Ff_2^U$ by inserting $0$'s into the coordinates corresponding to $V \setminus U$. Second, note that the choice of $U \supseteq V$ does not matter when defining $\widetilde{T}_1[L]$. Indeed, for a fixed $L$, the vertices $L \oplus H_U$ over all $U \supseteq V$ are in the same clique. Since $T_1$ is clique-consistent, it does not matter which $U$ is chosen when defining $\widetilde{T}_1[L]$, as all choices lead to the same function $T_1[L \oplus H_U]|_L$. Therefore the second prover can construct the table $\widetilde{T}_1$.

After constructing $\widetilde{T}_1$, the second prover then chooses a dimension $0 \leq r' \leq r$ uniformly for the advice $Q$. With probability at least $\frac{1}{r+1}$ the second prover also chooses $r' = r_1$. The second prover then chooses a zoom-out function pair $(\Ws, g_{Q, \Ws})$ that is 
$\left(\frac{C}{8 \cdot 5^{r}}, \frac{1}{5} \right)$-maximal with respect to $\widetilde{T}_1$ on $Q$ if one exists (and gives up otherwise).
Finally, the second prover extends the function $g_{Q, \Ws}$ randomly to a linear function on $\Ff_2^{V}$ to arrive at their answer. The resulting function is linear and it is equal to the inner product function $y \xrightarrow[]{} \langle s_{Q,V}, y \rangle$ for some unique string $s_{Q,V} \in \Ff_2^{V}$. The second prover outputs $s_{Q,V}$ as their answer.

\subsubsection{The Success Probability of the Provers}
In order to be successful, a series of events must occur. We go through each one and state the probability that each occurs. At the end this yields a lower bound on the provers' success probability. Recall that $J$ and $\beta$ are set according to~\eqref{eq: pcp parameters} in \cref{sec: covering properties} so that \cref{lm: covering} holds.

First, the provers need $U \in \Ug$, which occurs with probability at least $\frac{\epsilon}{2}$ by~\eqref{eq: good U}. Assuming that this occurs, the provers then both need to choose $r_1$ for the dimension of their zoom-in, which happens with probability at least $\left(\frac{1}{r+1} \right)^2 \geq \frac{1}{4r^2}$. In this case, both provers receive advice $Q$ as the span of $r_1$ random vectors chosen according to the distribution $\D'_{r_1}$.

The provers then need the advice to satisfy: $Q \in \Ql$, $Q \cap H_U = \{0\}$, and $Q \in \Qs$. When analyzing the probability that these three events occur, we need to recall that the advice vectors are actually drawn according to distribution $\mathcal{D}'_{r_1}$, the distribution described in \cref{sec: covering properties}. Fortunately, it is sufficient for us to analyze the probabilities under the uniform $\mc{D}_{r_1}$ and then apply \cref{lm: covering advice}.  By~\cref{th: consistent with side}, the first item occurs with probability at least $2^{-6\ell^2}$. On the other hand the probability that the second item does not occur is at most $\sum_{i=0}^{r_1} \frac{2^{i}2^J}{2^{3J}} \leq 2^{r_1+1 - 2J}$, while the probability that the third item does not occur is at most $\beta \cdot J^{1/4}$ by~\cref{lm: covering}. Altogether we get that with probability at least 
\[
2^{-6\ell^2}
-2^{r_1+1-2J}
-\beta \cdot J^{1/4} 
\geq 2^{-7\ell^2}
\]
under $\mathcal{D}_{r_1}$, we have $Q \in \Ql$, $Q \cap H_U = \{0\}$, and $Q \in \Qs$. By~\cref{lm: covering} we have that $Q \in \Ql$, $Q \cap H_U = \{0\}$, and $Q \in \Qs$ with probability at least $2^{-8\ell^2}$ under $\mathcal{D}_{r_1}'$ - the distribution which the $Q$ is actually drawn from.

Now let us assume that $U \in \Ug$, and both provers receive an $r_1$-dimensional advice $Q$ such that $Q \in \Ql$, $Q \cap H_U = \{0\}$, and $Q \in \Qs$. The first prover chooses the function $g_{Q,W_Q}: W_Q \xrightarrow[]{} \Ff_2$. Write $\codim(W_Q) = r_2$. Since $Q \in \Ql$,  we have by \eqref{eq: p1 consistency} that
\begin{equation} \label{eq: first prover dist}  
 \Pr_{L \subseteq \Ff_2^U}[g_{Q, W_Q}|_L \equiv T_1[L \oplus H_U]|_L \; | \; Q \subseteq L \subseteq W_Q] \geq C
\end{equation}
Since $Q \in \Qs$, we have
\begin{equation} \label{eq: D'_Q ineq}   
\E_{V} \left[\Pr_{L \subseteq \Ff_2^V}[g_{Q, W_Q}|_{L} \equiv T_1[L \oplus H_U]|_L \; | \; Q \subseteq L \subseteq W_Q \cap \Ff_2^V] \right] \geq \frac{C}{4}.
\end{equation}
Here we are using \cref{lm: covering} and the fact that the distributions over $L$ in \eqref{eq: first prover dist} and \eqref{eq: D'_Q ineq} are $\D_{Q, W_Q}$ and $\D'_{Q, W_Q}$ respectively.
Defining $W_Q[V] := W_Q \cap \Ff_2^V$, an averaging argument yields that with probability at least $\frac{C}{8}$, we have 
\begin{equation} \label{eq: V consistent}
\Pr_{L \subseteq \Ff_2^V}[g_{Q, W_Q}|_{L} \equiv T_1[L \oplus H_U]|_L \; | \; Q \subseteq L \subseteq W_Q[V]] \geq \frac{C}{8}.
\end{equation}

In addition, by~\cref{lm: retain codim}, we have that with probability at least $1 - 2^{r_2+3}\beta^2 J$, the codimension of $W_Q[V]$ inside of $\Ff_2^V$ is $r_2$. We call such $V$ consistent. Our argument above shows that $V$ is consistent with probability at least $\frac{C}{8} - 2^{r_2+3}\beta^2 J \geq \frac{C}{9}$

Suppose $V$ is consistent. In this case, the first prover's function, $g_{Q, W_Q}$, restricted to $W_Q[V]$ is also a candidate function in the second prover's table, $\widetilde{T}_1$. Define this function as $g_{Q, W_Q[V]}: \Ff_2^V \to \Ff_2$, given by $g_{Q, W_Q[V]} := g_{Q, W_Q}|_{W_Q[V]}$. By definition of $\widetilde{T}_1$, if $L \subseteq \Ff_2^V$ and $g_{Q, W_Q}|_{L} \equiv T_1[L \oplus H_U]|_L$, then we also have $g_{Q, W_Q[V]}|_{L} \equiv \widetilde{T}_1[L]$, so by \eqref{eq: V consistent}
\[
\Pr_{L \subseteq \Ff_2^V}[g_{Q, W_Q[V]}|_{L} \equiv \widetilde{T}_1[L] \; | \; Q \subseteq L \subseteq W_Q[V]] \geq \frac{C}{8}.
\]

From the definition of maximal-zoom outs, this means that there exists some zoom-out function pair, $(W'_Q[V], g_{Q, W'_Q[V]})$, that is $\left(\frac{C}{8 \cdot 5^{r^{2}}}, \frac{1}{5} \right)$-maximal and satisfies the following three properties: $W'_Q[V] \supseteq W_Q[V]$,  $g_{Q, W'_Q[V]}: W'_Q[V] \xrightarrow[]{} \Ff_2$ is linear, and $g_{Q, W'_Q[V]}|_{W[V]} = g_Q|_{W_Q[V]}$.  By~\cref{th: bounded zoom-out with zoom-in}, the number of $\left(\frac{C}{8 \cdot 5^{r}}, \frac{1}{5}\right)$-maximal zoom-out function pairs containing $Q$ that the second prover chooses from is at most 
\[
M \leq C^{-2} 2^{O_{k,\delta}(\ell)} = 2^{O_{k,\delta}(\ell)}.
\]

Thus, the second prover chooses $(W'_Q[V], g_{Q, W'_Q[V]})$ with probability at least $\frac{1}{M}$. Finally, if the second prover chooses $(W'_Q[V], g_{Q, W'_Q[V]})$, then the provers succeed if both provers extend their functions, $g_Q|_{W[V]}$ and $g_{Q, W'_Q[V]}$ in the same manner. This occurs with probability at least $2^{-\codim(W[V])} \geq 2^{-r_2}$. 

Putting everything together, we get that the provers succeed with probability at least
\[
\frac{\epsilon}{2} \cdot \frac{1}{4r^{2}} \cdot 2^{-8\ell^2} \cdot \frac{C}{9} \cdot \frac{1}{M} \cdot 2^{-r_2} = 2^{-c(\delta, k)\cdot \ell^2},
\]
for some constant $c(\delta, k) > 0$. In the above, the first term is the probability that $U \in \Ug$, the second term is the probability that both provers choose the same zoom-in dimension, the third term is the probability that $Q \in \Ql$, $Q \cap H_U = \{0\}$, $Q \in \Qs$, the fourth term is the probability that $V$ is consistent, the fifth term is the probability that the second prover chooses the a function that extends the first prover's answer, and the final term is the probability that both provers extend their functions in the same manner. This completes the proof of~\cref{lm: soundness outer to inner}. \hfill\qedsymbol

\section{Proof of \cref{thm: k dim matching}}
 \cref{thm: k dim matching} follows by plugging a modified version of our improved $k$-CSP hardness result into the reduction of \cite{LST}. Their reduction requires a degree $R$-bounded, fully-regular, $k$-partite, $k$-CSP as a starting point though, and our construction does not immediately have these properties. However, we can transform our $k$-CSP into one with these properties without compromising the parameters much. In the remainder of the section we make our $k$-CSP construction $k$-partite, then partwise-regular, and finally fully-regular, all while maintaining similar completeness and soundness. 

The first step, is to convert our weighted $k$-CSP, from \cref{thm: basic k query} into an unweighted instance. This can be done as in~\cite[Theorem 22]{crescenzi2001weighted}.\footnote{The result there discusses the ratio between the completeness and soundness, but it is easy to see that their reduction roughly preserves both. See~\cite[Theorem 4, Claim 23, Claim 24]{crescenzi2001weighted}.}

 % which works as follows. Suppose $\Psi$ is a hard instance of $k$-partite, $k$-CSP with alphabet size $R$, meaning the variables can be partitioned into $k$-parts such that each constraint contains exactly one variable per part. In \cite{LST}, it is shown that such a $\Psi$ can be randomly sparsified into a degree-$R$ bounded, $k$-partite $k$-CSP with alphabet size $R$, while maintaining a gap between completeness and soundness. Here, degree $R$-bounded means that each variable appears in at most $R$ constraints. This degree bounded CSP can thenbe  reduced to a $k$-Dimensional Matching instance with gap equal to that of the degree bounded CSP. Starting with our improved soundness in \cref{thm: main}, we end up with a better hardness of approximation factor for the final $k$-Dimensional Matching instance, but before we can apply the \cite{LST} reduction we need to transform our $k$-CSP in \cref{thm: main} to a $k$-partite one. This can be done generically while essentially keeping the same soundness as described below.

\subsection{Gaining $k$-partiteness}
A $k$-CSP is $k$-partite if its vertex set $V$ can be partitioned into parts $V_1, \ldots, V_k$ such that each constraint contains exactly one vertex in each part. In this case, we write the $k$-CSP as $\Psi = (V = V_1\cup \cdots \cup V_k, E, \Sigma_1, \ldots, \Sigma_k, \Phi = \{\Phi_e \}_{e \in E})$. Below we transform an arbitrary $k$-CSP into a $k$-partite $k$-CSP:
\begin{thm} \label{thm: k partite}
    For all $\eps, \delta > 0$ and arity $k$, and sufficiently large alphabet size $R$, it is NP-hard to distinguish between the following two cases given a $k$-partite, $k$-CSP $\Psi$ with alphabet size $R$: 
    % and redundancy at most $R^k$:
    \begin{itemize}
        \item YES case: $\val(\Psi) \geq 1 - \delta$,
        \item NO case: $\val(\Psi) \leq \frac{1}{R^{(1-\eps)(k-1)}}$.
    \end{itemize}
\end{thm}
\begin{proof}
Fix $\eps, \delta > 0$ and $k$. Let $\Psi$ to be the $k$-CSP with size $R$ alphabet $\Sigma$ and hypergraph $G = (V, E)$ from \cref{thm: main} with values $1-\delta$ and $\frac{1}{R^{(1-\eps/2) (k-1)}}$ in the YES and NO cases respectively. We modify $\Psi$ into a $k$-partite $\Psi'$ on the same alphabet $\Sigma$ as follows. Call the hypergraph $G' = (V', E')$. For each variable $v$ in $\Psi$, create copies $v(1), \ldots, v(k)$, and for each constraint $\Phi_e$ on edge $(u_1, \ldots, u_k)$ in $\Psi$, create constraints $\Phi_{e_{i_1, \ldots, i_k}}$ on the edges $e(u_1(i_1), u_2(i_2), \ldots, u_k(i_k))$ for each $k$-tuple of distinct indices $i_1, \ldots, i_k \in [k]$ in $\Psi'$. The constraint is defined as $\Phi'_{e_{i_1, \ldots, i_k}}(\sigma_1, \ldots, \sigma_k) = 1$ if and only if $\Phi_e(\sigma_1, \ldots, \sigma_k) = 1$ in the original CSP. We will show that this transformation does not change the value of $\Psi$ much.

    First suppose $\val(\Psi) \geq 1- \delta$ and let $A: V \to \Sigma$ be the labeling achieving this value. Then the labeling $A': V' \to \Sigma$ given by $A'(v(i)) = A(v)$ also achieves value $1-\delta$ in the $k$-CSP $\Psi'$, so $\val(\Psi') \geq 1- \delta$.

    Now suppose $\val(\Psi') = s$ and let $A'$ be the labeling achieving this value. Create a labeling $A: V \to \Sigma$ by choosing $A(v)$ to be one of the labels $A'(v(i))$ for $i \in [k]$ uniformly at random. For each $v$, call this randomly chosen $i$ the index of $v$. Then, choosing a random edge $(u_1, \ldots, u_k)$ in $E$, the indices $i_1, \ldots, i_k \in [k]$ of $u_1, \ldots, u_k$ are distinct with probability at least $\frac{k!}{k^k}$, and conditioned on this, $(u_1(i_1), \ldots, u_k(i_k))$ is a uniformly random constraint in $E'$. This constraint in $E'$ is satisfied by $A$ with probability at least $s$ and therefore, the randomly chosen $(u_1, \ldots, u_k) \in E$ is satisfied with probability at least $\frac{sk!}{k^k} \geq \frac{s}{2^k}$. It follows that if $\val(\Psi) \leq \frac{1}{R^{(1-\eps/2) (k-1)}}$, then $\val(\Psi') \leq \frac{2^k}{R^{(1-\eps/2)(k-1)}} \leq \frac{1}{R^{(1-\eps) (k-1)}}$. The last inequality holds for sufficiently large $R$.
\end{proof}

\subsection{Gaining Partwise-Regularity}

A $k$-partite, $k$-CSP, $\Psi = (V = V_1\cup \cdots \cup V_k, E, \Sigma_1, \ldots, \Sigma_k, \Phi = \{\Phi_e \}_{e \in E})$ is \emph{partwise-regular} if for each part $V_i$ each vertex $v_i \in V_i$ is contained in the same number of constraints. To make a $k$-partite, $k$-CSP partwise-regular, we use ideas related to the ``right-degree reduction'' technique from~\cite{MR,DinurHarsha}.

\begin{lemma}\label{lem:get_one_side}
    There is a polynomial time reduction that given a parameter $d\in\mathbb{N}$, a $k$-partite $k$-CSP instance $\Psi =  (V=V_1\cup\cdots\cup V_k,E,\Sigma_1,\ldots,\Sigma_k, \Phi=\{\Phi_e\}_{e\in E})$ in which the degree of each vertex is at least $d$ and the alphabet of each vertex has size at most $R$, and an index $i \in [k]$, produces a $k$-partite $k$-CSP instance $\Psi'= (V=V_1'\cup\ldots\cup V_k',E',\Sigma_1,\ldots,\Sigma_k, \Phi'=\{\Phi_e'\}_{e\in E'})$ with the following properties:
    \begin{enumerate}
        \item Completeness: if ${\sf val}(\Psi)\geq 1-\eps$, then ${\sf val}(\Psi')\geq 1-\eps$.
        \item Soundness: if ${\sf val}(\Psi)\leq \delta$, then 
        ${\sf val}(\Psi)\leq \delta+O\left(\sqrt{\frac{R^{k-1}}{d}}\right)$.
        \item Regularity: each vertex in $V_i$ has degree $d$.
        \item Retaining Regularity: for each $j\in [k]\setminus \{i\}$, if each vertex in $V_j$ has degree $d_j$, then each vertex in $V_j$ has degree $d_j d$.
    \end{enumerate}
\end{lemma}
\begin{proof}
    We will use the fact that for every pair of integers $N\geq D$, one may construct in polynomial time a bipartite expander graph with sides of size $N$, degrees equal to $D$ and second eigenvalue $O(1/\sqrt{D})$. See~\cite{alon2021explicit} in the case that $N$ is sufficiently larger than $D$, and a brute-force construction otherwise.

    Fix $\Psi$, $d$ and $i\in [k]$ as in the statement, and for each $v_i\in V_i$ denote by $d(v_i)\geq d$ the degree of $v_i$. For each $v_i\in V_i$ construct a bipartite expander graph $H_{v_i} = (L_{v_i}\cup R_{v_i}, E_{v_i})$ with $d(v_i)$ vertices on each side. We identify $L_{v_i}$ with the edges $e\in E$ incident to $v_i$, and we now explain the construction of $\Psi'$. For $j\neq i$ we set $V_j' = V_j$, and for $j=i$ we set $V_i' = V_i \times R_{v_i}$. In words, for each $v_i\in V_i$ we produce a cloud of copies of $v_i$ which we identify with $R_{v_i}$. The edge set $E'$ of $\Psi'$ consists of 
    $e' = (v_1,\ldots,v_{i-1},(v_i,u),v_{i+1},\ldots,v_k)$ such that $(v_1,\ldots,v_{i-1},v_{i+1},\ldots,v_k)\in L_{v_i}$, $u\in R_{v_i}$ and they are adjacent in $H_{v_i}$. The constraint on $e'$
    is the same as the constraint on $e = (v_1,\ldots,v_k)$. This completes the construction of $\Psi'$, and we move on to establishing its properties. It is clear that the vertices in each part $V'_i$ of $\Psi'$ have the same alphabet as the vertices in the corresponding part $V_i$ of $\Psi$, and that $\Psi'$ is a $k$-partite $k$-CSP, and we now move on to establishing each of the listed items.

    For the first item, given an assignment $A$ to $\Psi$ we produce an assignment $A'$ to $\Psi'$ where for $v_j'\in V_j'$ for $j\neq i$, we set $A'(v_j') = A(v_j')$. 
    For $(v_i,u)\in V_i'$ we set $A'(v_i,u) = A(v_i)$. It is easy to see that $A$ and $A'$ satisfy the same number of constraints.

    For the third item we note that the degree of $(v_i,u)\in V_i'$ is $d$, so the third item follows. 
    
    For the fourth item, fix $j\neq i$, suppose that each vertex in $V_j$ has degree $d_j$ and let $N(v_j, v_i)$ be the number of times that $v_j$ appears in a tuple in $L_{v_i}$ in the graph $H_{v_i}$. Then note that 
    $\sum\limits_{v_i\in V_i} N(v_j,v_i) = d_j$. For each $v_i\in V_i$, the number of constraints incident to both $v_j$ and some vertex in $\{v_i\}\times R_{v_i}$ is equal to     $N(v_j,v_i)\cdot d$. In total, we get that the degree of $v_j$ in $\Psi'$ is equal to $\sum\limits_{v_i\in V_i}N(v_j,v_i)\cdot d = dd_j$. It also follows that the number of edges in $\Psi'$ is $\card{E'} = d\card{E}$.

    We finish the proof by establishing the second item. Let $A_j'\colon V_j'\to\Sigma_j$ be a labeling for $\Psi'$ achieving ${\sf val}(\Psi')$. Define a labeling $A_j\colon V_j\to\Sigma_j$ to $\Psi$ by setting $A_j = A_j'$ for $j\neq i$. For $j=i$ and for each $v\in V_i$, we choose $u\in R_{v_i}$ uniformly and set $A_i(v_i) = A_i'(v_i,u)$.

    To analyze the construction, we denote $V_{-i} = \prod\limits_{j\neq i}V_j$, and denote by $v_{-i}$ elements in it. For each $v_i\in V_i$ and $\sigma_i\in \Sigma_i$, we define the sets $X_{v_i,\sigma_i}$ and $Y_{v_i,\sigma_i}$, as follows:
    \[
    X_{v_i,\sigma_i}
    = \{v_{-i}\in L_{v_i}
    ~|~
    (A_1(v_1),\ldots,A_{i-1}(v_{i-1}), \sigma_i,A_{i+1}(v_{i+1}),\ldots,A_k(v_k))\in \Phi_{(v_1,\ldots,v_k)}
    \},
    \]
    \[
    Y_{v_i,\sigma_i}
    =\{u\in R_{v_i}~|~A_i'(v_i,u) = \sigma_i\}.
    \]
   Then, the expected fraction of constraints satisfied by $(A_j')_{j=1,\ldots,k}$ is a
   lower bound on ${\sf val}(\Psi)$, and so
   \[
   {\sf val}(\Psi)
   \geq
   \frac{1}{\card{E}}
   \sum\limits_{v_{i}\in V_{i}}
   \sum\limits_{\sigma_i\in \Sigma_i}
   \frac{\card{Y_{v_i,\sigma_i}}}
   {\card{R_{v_i}}}
   \card{X_{v_i,\sigma_i}}
   =
   \frac{1}{\card{E}}
   \sum\limits_{v_{i}\in V_{i}}
   \sum\limits_{\sigma_i\in \Sigma_i}
   \frac{\card{Y_{v_i,\sigma_i}}}{d(v_i)}
   \card{X_{v_i,\sigma_i}}
   .
   \]
  Applying the expander mixing lemma in $H_{v_i}$, we get that
   \[
   \card{
   E(
   X_{v_i,\sigma_i},
   Y_{v_i,\sigma_i})-d
   \card{X_{v_i,\sigma_i}}\frac{\card{Y_{v_i,\sigma_i}}}{d(v_i)}}
   \leq O(\sqrt{d})\sqrt{\card{X_{v_i,\sigma_i}}\card{Y_{v_i,\sigma_i}}},
   \]
   and combining we get that
\[
   {\sf val}(\Psi)
   \geq
   \frac{1}{d\card{E}}
   \sum\limits_{v_{i}\in V_{i}}
   \sum\limits_{\sigma_i\in \Sigma_i}
    E(
   X_{v_i,\sigma_i},
   Y_{v_i,\sigma_i})
   -
   \frac{O(\sqrt{d})}{d\card{E}}
   \sum\limits_{v_{i}\in V_{i}}
   \sum\limits_{\sigma_i\in \Sigma_i}\sqrt{\card{X_{v_i,\sigma_i}}\card{Y_{v_i,\sigma_i}}}.
   \]
   For each $e' = (v_1,\ldots,v_{i-1},(v_i,u),v_{i+1},\ldots,v_k)\in E'$ satisfied by $A_1',\ldots,A_k'$, for 
   $\sigma_i = A_i'(v_i,u)$ we have that $v_{-i}\in X_{v_i,\sigma}$
   and $u\in Y_{v_i,\sigma_i}$, and then $(v_{-i}, (v_i,u))$ is an edge between $X_{v_i,\sigma}$ and 
   $Y_{v_i,\sigma}$.
   Thus, the first term on the right hand side is at least $\frac{1}{d\card{E}}\card{E'}{\sf val}(\Psi')={\sf val}(\Psi')$. 
   For the second term, by Cauchy-Schwarz
    \[
    \sum\limits_{v_{i}\in V_{i}}
   \sum\limits_{\sigma_i\in \Sigma_i}\sqrt{\card{X_{v_i,\sigma_i}}\card{Y_{v_i,\sigma_i}}}
   \leq
   \sum\limits_{v_{i}\in V_{i}}\sqrt{
   \sum\limits_{\sigma_i\in \Sigma_i}\card{X_{v_i,\sigma_i}}}
   \sqrt{
   \sum\limits_{\sigma_i\in \Sigma_i}\card{Y_{v_i,\sigma_i}}}.
    \]
    Fix $v_i\in V_i$. Since $Y_{v_i,\sigma}$ form a partition of $R_{v_i}$ we get that $\sum\limits_{\sigma_i}\card{Y_{v_i,\sigma_i}} = \card{R_{v_i}} = d(v_i)$. Also, since $\Psi$ has alphabet size $R$, for each $v_{-i}\in L_{v_i}$ there are at most $R^{k-1}$ values of $\sigma_i\in \Sigma_i$ such that $v_{-i}\in X_{v_i,\sigma_i}$, so
    $\sum\limits_{\sigma_i}\card{X_{v_i,\sigma_i}} \leq R^{k-1}\card{L_{v_i}} = R^{k-1}d(v_i)$. 
    Combining, we get that
    \[
        {\sf val}(\Psi)
        \geq 
        {\sf val}(\Psi')
        -
        O\left(\sqrt{dR^{k-1}}\right)
        \frac{1}{d\card{E}}
        \sum\limits_{v_i\in V_i}
        d(v_i)
        =
        {\sf val}(\Psi')
        -O\left(\sqrt{\frac{R^{k-1}}{d}}\right).\qedhere
    \]
\end{proof}
Iterating~\cref{lem:get_one_side} we get the following conclusion:
\begin{thm} \label{thm: k regular}
    For all $\eps, \delta > 0$ and arity $k$, and sufficiently large alphabet size $R$, it is NP-hard to distinguish between the following two cases given a partwise-regular, $k$-partite, $k$-CSP $\Psi$ with alphabet size $R$ :
    \begin{itemize}
        \item YES case: $\val(\Psi) \geq 1 - \delta$,
        \item NO case: $\val(\Psi) \leq \frac{1}{R^{(1-\eps)(k-1)}}$.
    \end{itemize}
\end{thm}
\begin{proof}
    We start with the instance 
    $\Psi$ in~\cref{thm: k partite} and set the parameter $d= k^2R^{3k}$. We duplicate each constraint in $\Psi$ $d$ times to ensure the degree of each vertex is at least $d$. We then apply~\cref{lem:get_one_side} $k$ times to get the result.
\end{proof}

\subsection{Gaining Full Regularity}
A $k$-CSP $\Psi$ is called fully-regular if each vertex is contained in the same number of constraints. Making a partwise-regular $k$-partite, $k$-CSP fully-regular is straightforward as well.

\begin{lemma} \label{lm: mult degrees}
There is a polynomial-time reduction that, given parameters $d_1, \ldots, d_k, c_1, \ldots, c_k \in \mathbb{N}$ and a $k$-partite, $k$-regular, $k$-CSP instance $\Psi = (V = V_1 \cup \cdots \cup V_k, E, \Sigma_1, \ldots, \Sigma_k, \Phi = \{\Phi_e \}_{e \in E})$ in which the degree of each vertex in $V_i$ is $d_i$, produces a fully-regular, $k$-partite $k$-CSP $\Psi' = (V' = V'_1 \cup \cdots \cup V'_k, E', \Sigma_1, \ldots, \Sigma_k, \Phi' = \{\Phi'_e \}_{e \in E'})$ such that $\val(\Psi) = \val(\Psi')$. The degree of a vertex in part $V'_i$ is $\frac{\prod_{j=1}^k c_jd_j}{c_i}$.
\end{lemma}
\begin{proof}
    We construct $\Psi'$ as follows. For each $i \in [k]$, set $V'_i = V_i \times [d_i] \times [c_i]$. In words, each vertex $v_i \in V_i$ of the original CSP gets duplicated into a cloud of $d_i \cdot c_i$ vertices. For the edges let $E'$ be the set of $((v_1, a_1, b_1), \ldots, (v_k, a_k, b_k)) \subseteq V'_1 \times \cdots \times V'_k$ such that $(v_1,\ldots, v_k) \in E$. The alphabet of each vertex in $V'_i$ is $\Sigma_i$ and for each edge $e' = ((v_1, a_1, b_1), \ldots, (v_k, a_k, b_k))$ the constraint $\Phi_{e'}$ is the same as $\Phi_e$ where $e = (v_1, \ldots, v_k) \in E$. It is clear that the degree of a vertex in part $V'_i$ is $\frac{\prod_{i=1}^k c_id_i}{c_i}$, so it remains to check that the value is preserved.

    We first argue that $\val(\Psi') \geq \val(\Psi)$. Let $A$ be an assignment to $V$ which satisfies $\val(\Psi)$ of the constraints. Then define the assignment $A'$ to $V'$ by assigning each vertex in the cloud of $v_i$ the label $A(v_i)$. Then note that for any $i$ and any vertices $v'_i \in V'_i$ and $v_i \in V_i$, we have
    \[
    \Pr_{e' \in E', v'_i \in e'}[e' \text{ is satisfied}] =  \Pr_{e \in E, v_i \in e}[e \text{ is satisfied}].
    \]
    Since in each CSP all of the vertices in $V'_i$ or $V_i$ have the same degree, this shows that $A'$ also satisfies $\val(\Psi)$ of the constraints in $\Psi'$, and as a result $\val(\Psi') \geq \val(\Psi)$

    To show $\val(\Psi') \leq \val(\Psi)$ note that we can partition $E'$ into edge sets of the form $E'_{a_1, b_1, \ldots, a_k, b_k}= \{((v_1, a_1, b_1), \ldots, (v_k, a_k, b_k)) \; | \; (v_1, \ldots, v_k) \in E\}$. Furthermore, for a fixed $a_1, b_1, \ldots, a_k, b_k$, if we restrict $\Psi'$ to only the edges $E'_{a_1, b_1, \ldots, a_k, b_k}$ and the vertices contained in these edges, we get a copy of $\Psi$. It follows that if an assignment to $\Psi'$ satisfies $\eta$-fraction of constraints, we can find a copy of $\Psi$ inside as well as an assignment to it satisfying at least $\eta$-fraction of the constraints. Thus, $\val(\Psi') \leq \val(\Psi)$.
\end{proof}

Putting all of the pieces together, we get a version of our main theorem for fully-regular, $k$-partite $k$-CSP.

\begin{thm} \label{thm: fully regular}
    For all $\eps, \delta > 0$ and arity $k$, and sufficiently large alphabet size $R$, it is NP-hard to distinguish between the following two cases given a fully-regular, $k$-partite,  $k$-CSP $\Psi$ with alphabet size $R$ :
    \begin{itemize}
        \item YES case: $\val(\Psi) \geq 1 - \delta$,
        \item NO case: $\val(\Psi) \leq \frac{1}{R^{(1-\eps)(k-1)}}$.
    \end{itemize}
\end{thm}
\begin{proof}
    We start with the instance 
    $\Psi$ in~\cref{thm: k regular}. Let $d_i$ be the degree of vertices in the part $V_i$. Apply \cref{lm: mult degrees} with $c_i = 1$.
\end{proof}
\subsection{Hardness of $k$-Dimensional Matching}
With \cref{thm: fully regular} in hand, we are ready to apply the reduction of \cite{LST}. 

\begin{thm} [\cite{LST}] \label{thm: lst sparsification}
     Let $0 \leq \lambda \leq 1$ and $C > 0$ be constants, fix an arity $k$, and let $R$ be a sufficiently large alphabet size. Then there is a randomized polynomial-time reduction which takes a fully-regular $k$-partite, $k$-CSP instance $\Psi$ with alphabet size $R$ and outputs a $R$-degree bounded $k$-CSP instance $\Psi'$ with alphabet size $R$, which satisfies the following with high probability:
     \begin{itemize}
         \item Completeness: $\val(\Psi') \geq \val(\Psi) - 3 \lambda$.
         \item Soundness: If $\val(\Psi) \leq CR^{-\gamma}$ for some $\gamma \geq 2$, then $\val(\Psi') \leq \frac{k(1+\lambda)}{(\gamma - 1)(1-\lambda)^2}$.
     \end{itemize}
 \end{thm}
 
\begin{thm}[\cite{LST}] \label{thm: matching reduction}
For sufficiently large prime $R$, there is a polynomial-time reduction from $R$-degree bounded, $k$-partite $k$-CSP, $\Psi$, with alphabet size $R$ to a $kR$-Dimensional Matching instance $\Pi$, such that $M \cdot \val(\Psi) = \val(\Pi)$, where $M$ is the number of constraints in $\Psi$.
\end{thm}
Combining \cref{thm: fully regular} and \cref{thm: lst sparsification}, we obtain the following.
 \begin{thm} \label{thm: bounded degree csp}
    Unless NP is contained in BPP, the following holds. For any $\eps > 0$, arity $k \geq 4$, and sufficiently large alphabet size $R$, there is no polynomial time algorithm which can distinguish the following two cases given an $R$-degree bounded $k$-CSP $\Psi$ with alphabet size $R$:
    \begin{itemize}
        \item YES Case: $\val(\Psi) \geq 1-\eps$.
        \item NO Case: $\val(\Psi) \leq \frac{k(1+\eps)}{(k-2)R}$
    \end{itemize}
 \end{thm}
 \begin{proof}
     Let $\Psi$ be an instance of $k$-CSP over alphabet size $R$ given by \cref{thm: k partite} where in the YES case $\val(\Psi) \geq 1- \eps/2$ and in the NO case $\val(\Psi) \leq R^{-(1-\eps/1000)(k-1)}$. By \cref{thm: k partite}, it is NP-hard to distinguish between these two cases. Now apply \cref{thm: lst sparsification} with $\lambda$ sufficiently small relative to $\eps$ and $\gamma = (1-\eps)(k-1)$ and let $\Psi'$ be the the $R$-degree bounded, alphabet size $R$, $k$-CSP that is outputted. If $\Psi$ satisfied the YES case, then with high probability $\val(\Psi') \geq 1 - \eps/2 - 3\lambda \geq 1-\eps$. If $\Psi$ satisfies the no case, then with high probability $\val(\Psi') \leq \frac{k(1+\lambda)}{((1-\eps/1000)(k-1)-1)(1-\lambda)^2} \leq \frac{k(1+\eps)}{(k-2)R}$. The result follows.
 \end{proof}
%Now combining \cref{thm: bounded degree csp} with \cref{thm: matching reduction}, we immediately obtain \cref{thm: k dim matching}.
 \begin{proof}[Proof of \cref{thm: k dim matching}]
     The theorem follows by combining \cref{thm: bounded degree csp} with the parameter $\eps$ there set sufficiently small and \cref{thm: matching reduction}. This yields an instance of $d$-Dimensional Matching with $d = kR$ and a hardness of approximation factor equal to $\frac{R(k-2)}{k}  \approx \frac{d(k-2)}{k^2}$. Setting $k=4$ achieves the desired result.
 \end{proof} 

\section*{Acknowledgments} We thank Venkat Guruswami for suggesting the problem.

\bibliographystyle{alpha}
\bibliography{references}

\newcommand{\etalchar}[1]{$^{#1}$}
\begin{thebibliography}{DKK{\etalchar{+}}21}

\bibitem[ALM{\etalchar{+}}98]{ALMSS}
Sanjeev Arora, Carsten Lund, Rajeev Motwani, Madhu Sudan, and Mario Szegedy.
\newblock Proof verification and the hardness of approximation problems.
\newblock {\em J. {ACM}}, 45(3):501--555, 1998.

\bibitem[Alo21]{alon2021explicit}
Noga Alon.
\newblock Explicit expanders of every degree and size.
\newblock {\em Combinatorica}, pages 1--17, 2021.

\bibitem[AS98]{AroraSafra}
Sanjeev Arora and Shmuel Safra.
\newblock Probabilistic checking of proofs: {A} new characterization of {NP}.
\newblock {\em J. {ACM}}, 45(1):70--122, 1998.

\bibitem[BDN17]{BDN}
Amey Bhangale, Irit Dinur, and Inbal~Livni Navon.
\newblock Cube vs. cube low degree test.
\newblock In {\em {ITCS} 2017}, volume~67 of {\em LIPIcs}, pages 40:1--40:31. Schloss Dagstuhl - Leibniz-Zentrum f{\"{u}}r Informatik, 2017.

\bibitem[BKS19]{BKS}
Boaz Barak, Pravesh~K. Kothari, and David Steurer.
\newblock Small-set expansion in shortcode graph and the 2-to-2 conjecture.
\newblock In {\em {ITCS} 2019}, volume 124 of {\em LIPIcs}, pages 9:1--9:12. Schloss Dagstuhl - Leibniz-Zentrum f{\"{u}}r Informatik, 2019.

\bibitem[Cha16]{Chan}
Siu~On Chan.
\newblock Approximation resistance from pairwise-independent subgroups.
\newblock {\em J. {ACM}}, 63(3):1--32, 2016.

\bibitem[CST01]{crescenzi2001weighted}
Pierluigi Crescenzi, Riccardo Silvestri, and Luca Trevisan.
\newblock On weighted vs unweighted versions of combinatorial optimization problems.
\newblock {\em Information and Computation}, 167(1):10--26, 2001.

\bibitem[DH13]{DinurHarsha}
Irit Dinur and Prahladh Harsha.
\newblock Composition of low-error 2-query {PCP}s using decodable {PCP}s.
\newblock {\em {SIAM} J. Comput.}, 42(6):2452--2486, 2013.

\bibitem[DKK{\etalchar{+}}18]{DKKMS1}
Irit Dinur, Subhash Khot, Guy Kindler, Dor Minzer, and Muli Safra.
\newblock Towards a proof of the 2-to-1 games conjecture?
\newblock In {\em {STOC} 2018}, pages 376--389. {ACM}, 2018.

\bibitem[DKK{\etalchar{+}}21]{DKKMS2}
Irit Dinur, Subhash Khot, Guy Kindler, Dor Minzer, and Muli Safra.
\newblock On non-optimally expanding sets in grassmann graphs.
\newblock {\em Israel Journal of Mathematics}, 243(1):377--420, 2021.

\bibitem[EKL24]{EvraKL}
Shai Evra, Guy Kindler, and Noam Lifshitz.
\newblock Polynomial bogolyubov for special linear groups via tensor rank, 2024.
\newblock arXiv:2404.00641.

\bibitem[FGL{\etalchar{+}}96]{FGLSS}
Uriel Feige, Shafi Goldwasser, L{\'{a}}szl{\'{o}} Lov{\'{a}}sz, Shmuel Safra, and Mario Szegedy.
\newblock Interactive proofs and the hardness of approximating cliques.
\newblock {\em J. {ACM}}, 43(2):268--292, 1996.

\bibitem[H{\aa}s01]{Hastad}
Johan H{\aa}stad.
\newblock Some optimal inapproximability results.
\newblock {\em J. {ACM}}, 48(4):798--859, 2001.

\bibitem[HKSS24]{HKSS}
Prahladh Harsha, Mrinal Kumar, Ramprasad Saptharishi, and Madhu Sudan.
\newblock An improved line-point low-degree test.
\newblock In {\em {FOCS} 2024}, pages 1883--1892. {IEEE}, 2024.

\bibitem[KMS17]{KMS}
Subhash Khot, Dor Minzer, and Muli Safra.
\newblock On independent sets, 2-to-2 games, and grassmann graphs.
\newblock In {\em {STOC} 2017}, pages 576--589. {ACM}, 2017.

\bibitem[KMS23]{KMS2}
Subhash Khot, Dor Minzer, and Muli Safra.
\newblock Pseudorandom sets in grassmann graph have near-perfect expansion.
\newblock {\em Annals of Mathematics}, 198(1):1--92, 2023.

\bibitem[KS13]{KS}
Subhash Khot and Muli Safra.
\newblock A two-prover one-round game with strong soundness.
\newblock {\em Theory Comput.}, 9:863--887, 2013.

\bibitem[KS15]{KSaket}
Subhash Khot and Rishi Saket.
\newblock Approximating {CSP}s using {LP} relaxation.
\newblock In {\em {ICALP} 2015}, volume 9134 of {\em Lecture Notes in Computer Science}, pages 822--833. Springer, 2015.

\bibitem[Lae14]{laekhanukit2014parameters}
Bundit Laekhanukit.
\newblock Parameters of two-prover-one-round game and the hardness of connectivity problems.
\newblock In {\em {SODA} 2014}, pages 1626--1643. {SIAM}, 2014.

\bibitem[LM24]{LeeManurangsi}
Euiwoong Lee and Pasin Manurangsi.
\newblock Hardness of approximating bounded-degree max 2-{CSP} and independent set on k-claw-free graphs.
\newblock In {\em {ITCS} 2024}, pages 71--1. Schloss Dagstuhl--Leibniz-Zentrum f{\"u}r Informatik, 2024.

\bibitem[LST25]{LST}
Euiwoong Lee, Ola Svensson, and Theophile Thiery.
\newblock Asymptotically optimal hardness for k-set packing and k-matroid intersection.
\newblock In {\em {STOC} 2025}, pages 54--61. {ACM}, 2025.

\bibitem[Min22]{dor_thesis}
Dor Minzer.
\newblock {\em On Monotonicity Testing and the 2-to-2 Games Conjecture}, volume~49 of {\em {ACM} Books}.
\newblock {ACM}, 2022.

\bibitem[MNT22]{MNT}
Pasin Manurangsi, Preetum Nakkiran, and Luca Trevisan.
\newblock Near-optimal np-hardness of approximating max $k$-{CSP}$_r$.
\newblock {\em Adv. Math. Commun.}, 18:1--29, 2022.

\bibitem[MR10]{MR}
Dana Moshkovitz and Ran Raz.
\newblock Two-query {PCP} with subconstant error.
\newblock {\em J. {ACM}}, 57(5):29:1--29:29, 2010.

\bibitem[MZ23]{MZ}
Dor Minzer and Kai Zheng.
\newblock Approaching the soundness barrier: {A} near optimal analysis of the cube versus cube test.
\newblock In {\em {SODA} 2023}, pages 2761--2776. {SIAM}, 2023.

\bibitem[MZ24]{MZ24}
Dor Minzer and Kai~Zhe Zheng.
\newblock Near optimal alphabet-soundness tradeoff {PCP}s.
\newblock In {\em {STOC} 2024}, pages 15--23. {ACM}, 2024.

\bibitem[Rao11]{Rao}
Anup Rao.
\newblock Parallel repetition in projection games and a concentration bound.
\newblock {\em {SIAM} J. Comput.}, 40(6):1871--1891, 2011.

\bibitem[Raz98]{Raz}
Ran Raz.
\newblock A parallel repetition theorem.
\newblock {\em {SIAM} J. Comput.}, 27(3):763--803, 1998.

\bibitem[RS97]{RS}
Ran Raz and Shmuel Safra.
\newblock A sub-constant error-probability low-degree test, and a subconstant error-probability {PCP} characterization of {NP}.
\newblock In {\em {STOC} 1997}, pages 475--484, 1997.

\end{thebibliography}
\appendix
\end{document}